\documentclass{amsart}
\usepackage{hyperref}
\usepackage{graphicx}
\usepackage{curves}
\usepackage{enumerate}
\allowdisplaybreaks

\newcommand*{\mailto}[1]{\href{mailto:#1}{\nolinkurl{#1}}}

\newtheorem{theorem}{Theorem}[section]
\newtheorem{lemma}[theorem]{Lemma}

\newtheorem{remark}[theorem]{Remark}

\newcommand{\R}{\mathbb{R}}
\newcommand{\Z}{\mathbb{Z}}

\newcommand{\C}{\mathbb{C}}

\newcommand{\M}{\mathbb{M}}

\newcommand{\nn}{\nonumber}
\newcommand{\be}{\begin{equation}}
\newcommand{\ee}{\end{equation}}
\newcommand{\bea}{\begin{eqnarray}}
\newcommand{\eea}{\end{eqnarray}}
\newcommand{\ul}{\underline}
\newcommand{\ol}{\overline}
\newcommand{\pa}{\partial}

\newcommand{\id}{\mathbb{I}}
\newcommand{\I}{\mathrm{i}}
\newcommand{\E}{\mathrm{e}}

\newcommand{\clos}{\mathop{\mathrm{clos}}}
\newcommand{\re}{\mathop{\mathrm{Re}}}
\newcommand{\im}{\mathop{\mathrm{Im}}}

\DeclareMathOperator{\res}{Res}

 \newcommand{\noprint}[1]{}

\newcommand{\si}{\sigma}
\newcommand{\la}{\lambda}

\numberwithin{equation}{section}

\newcommand{\sigI}{\begin{pmatrix} 0 & 1 \\ 1 & 0 \end{pmatrix}}

\begin{document}

\title[Long-Time Asymptotics for the Toda Shock Problem]{Long-Time Asymptotics for the Toda Shock Problem: Non-Overlapping Spectra}

\author[I. Egorova]{Iryna Egorova}
\address{B. Verkin  Institute for Low Temperature Physics\\ 47,Nauky ave\\ 61103 Kharkiv\\ Ukraine}
\email{\href{mailto:iraegorova@gmail.com}{iraegorova@gmail.com}}

\author[J. Michor]{Johanna Michor}
\address{Faculty of Mathematics\\ University of Vienna\\
Oskar-Morgenstern-Platz 1\\ 1090 Wien\\ Austria}
\email{\href{mailto:Johanna.Michor@univie.ac.at}{Johanna.Michor@univie.ac.at}}
\urladdr{\href{http://www.mat.univie.ac.at/~jmichor/}{http://www.mat.univie.ac.at/\string~jmichor/}}

\author[G. Teschl]{Gerald Teschl}
\address{Faculty of Mathematics\\ University of Vienna\\
Oskar-Morgenstern-Platz 1\\ 1090 Wien}
\email{\href{mailto:Gerald.Teschl@univie.ac.at}{Gerald.Teschl@univie.ac.at}}
\urladdr{\href{http://www.mat.univie.ac.at/~gerald/}{http://www.mat.univie.ac.at/\string~gerald/}}

\dedicatory{To Vladimir Aleksandrovich Marchenko with deep admiration on the occasion of his 95th birthday}

\thanks{Research supported by the Austrian Science Fund (FWF) under Grants No.\ Y330, V120,
and by the grant "Network of Mathematical Research 2013--2015"}
\thanks{Zh. Mat. Fiz. Anal. Geom. {\bf 14}, 406--451 (2018)}

\begin{abstract}
We derive the long-time asymptotics for the Toda shock problem using the nonlinear steepest descent
analysis for oscillatory Riemann--Hilbert factorization problems. We show that the half-plane of space/time variables
splits into five main regions: The two regions far outside where the solution is close to the free backgrounds.
The middle region, where the solution can be asymptotically described by a two band solution, and two regions
separating them, where the solution is asymptotically given by a slowly modulated two band
solution. In particular, the form of this solution in the separating regions verifies a conjecture from
Venakides, Deift, and Oba from 1991.			
\end{abstract}

\maketitle

\section{Introduction}
The investigation of shock waves in the Toda lattice goes back at least to the numerical works of Holian and Straub \cite{hs} and Holian, Flaschka, and McLaughlin \cite{hfm}.
A theoretical investigation was later on done by Venakides, Deift, and Oba \cite{vdo} employing the Lax--Levermore method. As their main result they showed (in the case of some special
symmetric initial conditions) that in a sector $|\frac{n}{t}|<\xi^\prime_{cr}$ the solution can be asymptotically described by a period two solution, while in a sector $|\frac{n}{t}|>\xi_{cr}$
the particles are close to the unperturbed lattice. For the remaining region $\xi^\prime_{cr}<|\frac{n}{t}|<\xi_{cr}$ the solution was conjectured to be asymptotically close to a
modulated single-phase quasi-periodic solution but this case was not solved there. Despite some follow-up publications by Bloch and Kodama \cite{bk,bk2} and
Kamvissis \cite{km0} this problem remained open. The aim of the present paper is to fill this gap. Our method of choice will be the formulation of the inverse scattering problem
as a Riemann--Hilbert problem and an application of the nonlinear steepest descent analysis developed by Deift and Zhou \cite{dz} based on earlier ideas from Manakov \cite{ma}
and Its \cite{its}. For more on its history and an overview of this method applied to the Toda lattice in the classical case of constant background we refer to \cite{KTb} (cf. also \cite{km,KTc}) and the
references therein. Soon after the introduction of this method Deift, Kamvissis, Kriecherbauer, and Zhou \cite{dkkz} applied it to another steplike situation, the Toda rarefaction problem.
However, only the case $t\to\infty$ with $n$ fixed was considered there. In fact, asymptotics in the $(n,t)$ plane require an extension of the original nonlinear steepest descent analysis
based on a suitably chosen $g$-function as first introduced in Deift, Venakides, and Zhou \cite{dvz}. Recently this was done for the modified Korteweg--de Vries equation by
Kotlyarov and Minakov \cite{KM,KM2,mi} and for the Korteweg--de Vries equation by two of us jointly with Gladka and Kotlyarov \cite{egkt}. However, all these works have in
common that the spectra of the underlying Lax operators overlap and hence the associated Riemann surface is simply connected. While Riemann--Hilbert problems on
nontrivial Riemann surfaces have a long tradition, see e.g.\ the monograph by Rodin \cite{ro}, the nonlinear steepest descent analysis in such situations was developed only
recently by Kamvissis and one of us \cite{kt,kt2} (see also \cite{KTc,MT}). It is our main novel feature in the present paper to formulate the problem on a Riemann surface
formed by combining both spectra and working on this surface. More precisely, in the most interesting region $\xi_{cr,1}<|\frac{n}{t}|<\xi_{cr}$ we will work on a dynamically
adapted surface.

To describe our results in more detail we recall that the Toda shock problem consists of studying the long-time asymptotics of solutions of the doubly infinite Toda lattice
\begin{align} \label{tl}
	\begin{split}
\dot b(n,t) &= 2(a(n,t)^2 -a(n-1,t)^2),\\
\dot a(n,t) &= a(n,t) (b(n+1,t) -b(n,t)),
\end{split} \quad (n,t) \in \Z \times \R,
\end{align}
with so called steplike shock initial profile
\begin{align} \label{ini1}
\begin{split}	
& a(n,0)\to a_1, \quad b(n,0) \to b_1, \quad \mbox{as $n \to -\infty$}, \\
& a(n,0)\to a,\  \quad b(n,0) \to b, \quad \ \mbox{as $n \to +\infty$},
\end{split}
\end{align}
where the background Jacobi operators with constant coefficients
\begin{align} \label{H1} 
\begin{split}
(H y)(n)& = a y(n-1) +by(n) + a y(n+1), \\ 
(H_1 y)(n)& =a_1 y(n-1) +b_1y(n) + a_1 y(n+1),
\end{split} \quad n\in\Z,
\end{align}
have spectra with the following mutual location: $\inf\si(H_1)<\inf\si(H)$. These spectra can either overlap or not, and it produces essentially different types  of asymptotical behavior of the solution.

For the steplike case in the general situation $\si(H_1)\neq \si(H)$
there are two principal cases distinguished by the conditions $\inf \si(H_1)<\inf \si(H)$ (the Toda shock problem) and $\inf \si(H_1)>\inf \si(H)$ (the Toda rarefaction problem).
As mentioned before, the Toda shock problem was studied partly in \cite{vdo} for non-overlapping background spectra of equal length, and the Toda rarefaction problem in \cite{dkkz} using the Riemann--Hilbert problem approach
for finite $n$ only, as $t\to \infty$, under the restriction that the spectra are again equal in length, non-overlapping, and that the discrete spectrum is symmetric with respect to $0$.
An overview on the asymptotic solution in the general situation can be found in \cite{m16}.

In this study we analyze the asymptotical behavior of the solution of the Toda shock problem in the space-time half-plane
$(n,t)\in\Z\times \R_+$ in the case of arbitrary non-overlapping background spectra $\sup\si(H_1)<\inf\si(H)$. 
For $t<0$, the lattice behaves as a solution of the so called Toda rarefaction problem and will be considered in a forthcoming paper. 
We consider the value $\xi:=\frac{n}{t}$ as a {\it slow variable} and propose the precise form of the solution in a vicinity of the rays $\xi=const$ as usual.
We only compute the leading terms of the long-time asymptotics of the solutions, but in all principal regions of the space-time half-plane, excluding small transition regions.
To simplify our exposition, we assume that no eigenvalues are present in the domain $\R \setminus [\inf \sigma(H_1), \sup \sigma(H)]$. They can easily be added using the techniques developed in \cite{KTc}. We suppose that there is one eigenvalue in the gap $(\sup \sigma(H_1), \inf \sigma(H))$ to compare our result with the results of \cite{vdo}. 
We will also not provide detailed error estimates or study the case of overlapping spectra but defer these to forthcoming papers.

\begin{figure}[ht]
\centering
\includegraphics[width=12cm]{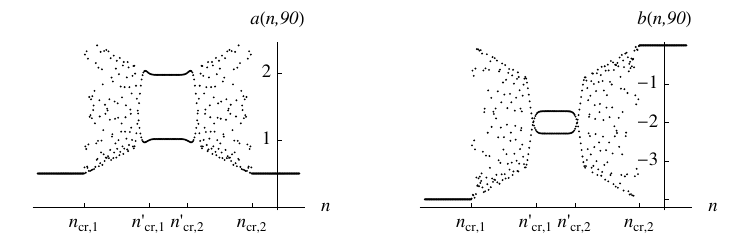}
\caption{The case for a pure step with $\sigma(H_1)=[-5,-3]$ and $\sigma(H)=[-1,1]$, without 
discrete spectrum.}\label{fig1}
\end{figure}

In Fig.~\ref{fig1} the numerically computed solution corresponding to the initial 
condition $a_1=1/2$, $b_1=-4$, $a=1/2$, and $b=0$ is shown. The left picture depicts the 
function $n \mapsto a(n,t)$ at a frozen time $t=90$. In areas where the function 
seems to be continuous this is due to the fact that we have plotted a large number 
of particles (around $800$) and also due to the $2$-periodicity in space. So one can
think of the two lines in the middle region as the even- and odd-numbered particles of the lattice.
 
Let us give a short qualitative description of our result. 
There are five principal regions on the half plane $(n,t)$ divided by rays $n/t = \tilde \xi$, with $\tilde \xi = \xi_{cr,1}$,$ \xi_{cr,1}^\prime, \xi_{cr, 0}$, $\xi_{cr}^\prime, \xi_{cr}$ where $\xi_{cr,1}<\xi_{cr,1}^\prime< \xi_{cr, 0}  <\xi_{cr}^\prime<\xi_{cr}$. 
  In the domain $\xi>\xi_{cr}$, the solution is asymptotically close to the constant right background solution $\{a, b\}$, and in the domain $\xi<\xi_{cr,1}$ it is close to the left background $\{a_1, b_1\}.$ 
In the domain $\xi_{cr}^\prime< \xi <\xi_{cr}$, there 
appears a monotonous smooth function $\gamma(\xi) \in \R$ such that
 $\gamma(\xi_{cr}^\prime)=\sup \sigma(H_1)$, $\gamma(\xi_{cr})=\inf \sigma(H_1)$.
When the parameter $\xi$ starts to decay from the point $\xi_{cr}$, the point $\gamma(\xi)$ ``opens'' a band 
$[\inf \sigma(H_1),\gamma(\xi)]$ (the Whitham zone). This interval  and $\sigma(H)$ can be treated as the bands of a (slowly modulated) two band solution of the Toda lattice, which turns out to give the leading asymptotical term of our solution with respect to large $t$. This two band solution is defined uniquely by its initial divisor. We compute this divisor precisely via the values of the right transmission coefficient on the interval $[\inf \sigma(H_1),\gamma(\xi)]$ (see formulas \eqref{divizor}, \eqref{Deltaj}, \eqref{import6}, \eqref{aee}, and \eqref{bee} below). Thus, in a vicinity of any  ray $\frac{n}{t}=\xi$ the solution of \eqref{tl}--\eqref{ini1} is asymptotically finite-gap (Theorem~\ref{theor1}).
This asymptotical term also can be treated as a function of $n$, $t$, and $\frac{n}{t}$ in the whole domain $t(\xi_{cr}^\prime+\varepsilon)<n<t(\xi_{cr}-\varepsilon)$. 
A numerical comparison between the solution and the corresponding asymptotic formula in this region is shown in Fig.~\ref{fig2}.
\begin{figure}[ht]
\centering
\includegraphics[width=12cm]{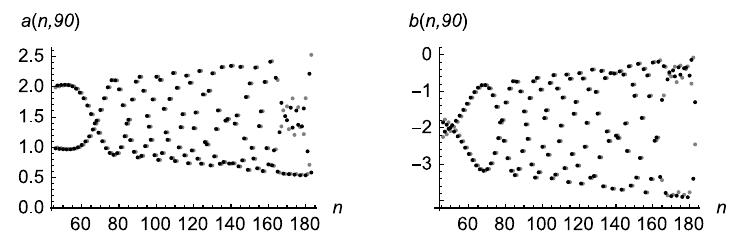}
\caption{Comparison between the solution (black) and the asymptotic formula (gray) in the region $\xi_{cr}^\prime<\xi<\xi_{cr}$.}\label{fig2}
\end{figure}
Next, in the domains $\xi_{cr, 0} < \xi < \xi_{cr}^\prime$ and $\xi_{cr,1}^\prime < \xi < \xi_{cr, 0}$, 
the asymptotic of the solution of \eqref{tl}--\eqref{ini1} is described by two finite-gap solutions. They are 
connected with one and the same intervals $\sigma(H_1)$ and $\sigma(H)$ and the initial divisors (or shifts of the phase) do not depend on the slow variable $\xi$, but differ due to the presence of the soliton. The situation in the 
domain $\xi_{cr,1}<\xi < \xi_{cr,1}^\prime$ is similar to the Whitham zone described above. There 
appears a monotonous smooth function $\gamma_1(\xi) \in \R$ such that $\gamma_1(\xi_{cr,1})=\sup \sigma(H)$, $\gamma_1(\xi_{cr,1}^\prime)=\inf \sigma(H)$.
The finite-gap asymptotic here is again local along the ray, and is defined by the intervals $\sigma(H_1)$ and 
$[\gamma_1(\xi), \sup \sigma(H)]$.

We do not study the transitional regions in vicinities of the points $\inf \sigma(H_1)$ and $\sup \sigma(H)$, but one can expect the appearance of asymptotical solitons here (see \cite{bdmek}). We emphasize that the RH problem with jumps on several disjoint intervals was first treated rigorously in \cite{deift}.  Among the results of this 
seminal paper was a formula for the leading term of the asymptotics for coefficients of the respective Jacobi matrix, given in terms of a quotient of theta functions. 
However, from that formula it was hard to see that it is a finite-gap Jacobi operator, 
which was later shown in \cite{ep}. In contradistinction to the situation considered in \cite{deift, ep}, in the present paper 
the jump contour of the limiting RH problem depends on the variable $\xi$, which dictates a special choice 
of the $g$-functions to replace the phase functions. We choose these $g$-functions as linear combinations of Abel integrals
of the second and third kind such that we can easily control the lines where $\re g=0$
(cf.\ Fig.~ \ref{fig:Reg} below). To study the specific properties of the $g$-functions in detail it is convenient to use standard properties of the associated Riemann surface, which pushed us to consider the RH problem on 
the Riemann surface. We  emphasize that this approach leads to quite simple and natural asymptotic 
formulas for the solution of 
\eqref{tl}--\eqref{ini1}, which are exact finite-gap solutions of the Toda lattice 
considered in a small vicinity of the ray $\xi=const$.

\section{Statement of the Riemann--Hilbert problems}

To set the stage we describe the class of initial data which we study. Without loss of generality, by shifting and scaling of the spectral parameter of the Jacobi spectral equation
\be \label{jse}
(\mathcal H(t)\psi)(n) :=a(n-1,t)\psi(n-1) + b(n,t)\psi(n) +a(n,t)\psi(n+1)=\la\psi(n) 
\ee
we can reduce the asymptotics of the initial data  to
\begin{align} \label{ini}
\begin{split}	
& a(n,0)\to \frac{1}{2}, \quad b(n,0) \to 0, \quad \mbox{as $n \to +\infty$}, \\
& a(n,0)\to d, \quad b(n,0) \to - c, \quad \mbox{as $n \to -\infty$},
\end{split}
\end{align}
where $c,d \in \R_+$ are  constants satisfying the conditions
\be\label{cond7}
c>1,\quad 0<2d < c-1.
\ee
The spectra of the free (background) Jacobi operators \eqref{H1} are given now by $\si(H_1)=[-c-2d,-c+2d]$ and $\si(H)=[-1,1]$. 
We suppose that the initial data decay to their backgrounds exponentially fast
\be \label{decay}
 \sum_{n = 0}^{\infty} \E^{V   n} \Big(|a(-n,0) - d| + |b(-n,0)+c|+|a(n,0) - \frac{1}{2}| + |b(n,0)|\Big) < \infty,
\ee
where for a small positive $\varepsilon$
\be\label{cont12}\cosh V = \max\left\{c+2d ,\,(1+c)(2d)^{-1}\right\} +2\varepsilon.
\ee
Let $a(n,t)$, $b(n,t)$ be the unique solution of the Cauchy problem \eqref{tl} with initial condition of the type \eqref{ini}--\eqref{cont12}. It is known (\cite[Lemma 3.2]{emt3}, \cite{Teschl2}) that the decay condition \eqref{decay} is preserved by the time evolution
of the Toda lattice, and therefore for any fixed $t$ the solution $a(n,t)$, $b(n,t)$ is exponentially close to the background constant asymptotics as $n\to\pm\infty$.

The spectrum of the Jacobi operator $\mathcal H(t)$ consists of an (absolutely) continuous part
$[-c-2d,-c+2d] \cup [-1, 1]$ of two nonintersecting bands of spectra of multiplicity one,
plus possibly a finite number of eigenvalues. 
For simplicity we assume in addition to \eqref{decay} and \eqref{cond7} that
\be \label{cond15}
 \mbox{the discrete spectrum of $\mathcal H(0)$ consists of one point $\lambda_0 \in (-c+2d,-1)$}
\ee
such that we can compare our results with \cite{vdo}, where a single soliton is present. It is easy to extend our result
to an arbitrary finite number of eigenvalues using standard techniques \cite{KTa}.

\subsection{Elements of scattering theory}
In this paper we will use either left or right scattering data of the operator $\mathcal H(t)$, depending on which region of the
space-time half plane we investigate, and apply the Riemann--Hilbert (RH) problem approach in 
vector form (cf.\ \cite{dkkz}). To this end we recall some facts from
scattering theory of Jacobi operators with steplike backgrounds from \cite{emtstp2}. 
Instead of the complex plane with a cut along the continuous spectrum consider the spectral data 
of $\mathcal H(t)$ on the upper sheet of the Riemann surface $\M$ connected with the function
\be \label{R1/2}
R^{1/2}(\la)= -\sqrt{(\la^2-1)\big((\la+c)^2-4d^2)\big)},
\ee
where $\sqrt{\la}=|\sqrt{\la}|\E^{\frac{\I \arg(\la)}{2}}$, $-\pi<\arg(\la)<\pi$, is the standard root with branch cut along $(-\infty,0]$.
A point on $\M$ is denoted by $p=(\la, \pm)$, $\la \in \C$, with $(\infty, \pm): = \infty_\pm$.
The projection onto $\C \cup \{\infty\}$ is denoted by $\pi(p)=\la$.
The sheet exchange map is given by
$
p^*=
(\la,\mp)  \text{ for } p=(\la,\pm).$
The sets
\begin{align*}
\Pi_U = \{(\la,+) \mid \la \in \C \setminus (\sigma(H_1) \cup \sigma(H))\} \subset \M, \quad
\Pi_L = \{p^* \mid p \in \Pi_U\},
\end{align*}
are called upper, lower sheet, respectively. Denote
\be\aligned \label{siul}
\Sigma_{u}&=\{p=(\la + \I 0,+)\},\quad \Sigma_{\ell}=\{p=(\la - \I 0,+)\}, \quad\la \in \sigma(H),\\
\Sigma_{1,u}&=\{p=(\la + \I 0,+)\},\quad \Sigma_{1,\ell}=\{p=(\la - \I 0,+)\}, \quad\la \in \sigma(H_1),\endaligned
\ee
and $\Sigma=\Sigma_{u}\cup \Sigma_{\ell}$,   $\Sigma_1=\Sigma_{1,u}\cup \Sigma_{1,\ell}$. We consider $\Sigma$ and $\Sigma_1$
as clockwise oriented contours, when looking on the upper sheet. For any function $f(p)$ holomorphic in a neighborhood of $ \Gamma:=\Sigma_1 \cup \Sigma$ on $\Pi_U$
and continuous up to the boundary, we consider its value on the contour as 
\be \label{fpprime}
f(p)= \lim_{p^\prime \in \Pi_U \to p} f(p^\prime), \quad p \in \Gamma.
\ee
The points $p=(\la + \I 0,+)$ and $\overline p=(\la - \I 0,+)$ are called symmetric points of $\Gamma$.

On $\Pi_U$, introduce two new spectral variables $z(p)$ and $z_1(p)$, with $|z(p)|<1$ and $|z_1(p)|<1$, by
\be \label{z_1}
z(p) = \la - \sqrt{\la^2 - 1}, \quad z_1(p)= \frac{1}{2d}\big(\la + c - \sqrt{(\la+c)^2-4d^2}\big).
\ee
These variables are different Joukovski transformations of the spectral parameter
\[
\la=\frac{1}{2}\left(z + z^{-1}\right)= - c + d\left(z_1 + z_1^{-1}\right).
\]
The functions $y(\la, n)=z(p)^n$ and $y_1(\la,n)=z_1(p)^{-n}$
are the ``free exponents'' connected to the background operators $H$ and $H_1$, respectively.

On $\clos\Pi_U:=\Pi_U \cup \Gamma$, there exist Jost solutions $\psi(p,n,t)$ and $\psi_1(p,n,t)$ of the equation
\be\label{imp90}\mathcal H(t)\phi(p,n,t)=p\, \phi(p,n,t),\quad p\in \clos\Pi_U,
\ee
which asymptotically look like the free solutions of the background equations,
\[
\lim_{n\to \infty} z^{-n}(p)\psi(p,n,t) =1, \quad
\lim_{n\to -\infty} z_1^{n}(p)\psi_1(p,n,t) =1, \quad p \in \clos\Pi_U.
\]
These solutions  satisfy 
\be\label{SP1}\aligned \psi(\ol p, n, t)&=\ol{\psi(p,n,t)},\ \ p\in\Sigma;\quad \ \psi(p,n,t)\in\R,\ \ p\in\R\setminus \sigma(H), \\
\psi_1(\ol p, n, t)&=\ol{\psi_1(p,n,t)},\  p\in\Sigma_1;\quad \psi_1(p,n,t)\in\R,\ \ p\in\R\setminus \sigma(H_1).\endaligned\ee They can be represented via the transformation operators
\[\psi_1(p,n,t)=\sum_{m=n}^{-\infty}K_1(n,m,t)z_1(p)^{-m},\quad \psi(p,n,t)=\sum_{m=n}^{+\infty}K(n,m,t)z(p)^{m},\]
where the real-valued functions $K(n,m)$ and $K_1(n,m)$ satisfy due to \eqref{decay}--\eqref{cont12} 
\begin{align} \label{decay3} 
& K_1(n,m,t)\leq C_1(n,t)\E^{\frac{V(n+m)}{2}},\ m<n;\ \  K(n,m,t)\leq C_2(n,t)\E^{ - \frac{V(n+m)}{2}},\  m>n, \\
\label{decay5}
& K_1(n,n,t)=1 +O(\E^{Vn}), \   n\to -\infty; \
K(n,n,t)=1 +O(\E^{-Vn}), \  n\to +\infty.
\end{align}
Introduce two  values $\rho_j>1$, $j=1,2$, such that
\be \label{constabs}\rho_1 +\rho_1^{-1}=(d)^{-1}(1+c)+2\varepsilon,\quad \rho_2+\rho_2^{-1}=2(c + 2d +\varepsilon).\ee
Let $\mathfrak D$ be a domain in $\Pi_U$ defined by
\be\label{defel}
\mathfrak D=\{p\in\Pi_U: \ 1<|z_1^{-1}(p)|<\rho_1,\  1<|z^{-1}(p)|<\rho_2\}.
\ee
The constants $\rho_j$ are chosen in such a way that each pre-image of the circles $|z(p)|=\rho_2^{-1}$ 
and $|z_1(p)|<\rho_1^{-1}$ on $\Pi_U$, which is an ellipse, contains both $\Sigma$ and $\Sigma_1$. 
On the other hand, $\rho_j<V$, respectively,
$|z^{-1}(p)|<V$, $|z_1^{-1}(p)|<V$ uniformly in $\mathfrak D$. Thus one can introduce a solution of \eqref{imp90}, which is an analytical continuation  of $\ol{\psi(p,n,t)}$ (resp.\ 
$\ol{\psi_1(p,n,t)}$) to $\mathfrak D$, usually defined on $\Sigma$ (resp.\ $\Sigma_1$)
\[
\breve\psi_1(p,n,t)=\sum_{m=n}^{-\infty}K_1(n,m,t)z_1(p)^{m},\quad \breve\psi(p,n,t)=\sum_{m=n}^{+\infty}K(n,m,t)z(p)^{-m}.
\]
From \eqref{decay3} and \eqref{decay5} it follows that  
\be\label{cont66}\langle\psi_1, \breve\psi_1\rangle(p,t)=\sqrt{(p-c)^2 - 4d^2},
\quad \langle\psi, \breve\psi\rangle(p,t)=-\sqrt{p^2 - 1}, \quad p\in\mathfrak D,
\ee
where
$\langle f,g\rangle (p,t)=a(n-1,t)(f(p,n-1,t)g(p,n,t) - f(p,n,t)g(p,n-1,t))$ is the Wronskian of two solutions of 
\eqref{jse}. Denote by $W(p,t)=\langle \psi_1,\psi\rangle(p,t)$
the Wronskian of the Jost solutions. By \eqref{cond15}, $W(p,t)$ has on $\Pi_U$ the only simple zero at  
$p_0=(\la_0,+)$ and does not vanish on $\Sigma$ except at possibly the edges of the continuous spectrum 
$\pa\si:=\{-c-2d,-c+2d\}\cup \{-1,1\}=\pa\si(H_1)\cup\pa\si(H).$ If 
$W(E,t)=0$ for $E\in\pa\sigma$, we call the point $E$ a {\it resonant point}. If $E$ is a resonant point then $W(p,t)=C(t)\sqrt{p-E} (1+o(1))$ as $p\to E$, with $C(t)\neq 0$ for all $t\in\R_+$.

The Jost solutions  satisfy the scattering relations
\begin{align}
  \label{rsr}
  T(p,t) \psi_1(p,n,t)
  & =\overline{\psi(p,n,t)} + R(p,t)\psi(p,n,t), \quad p \in \Sigma, \\  \label{lsr}
  T_1(p,t) \psi(p,n,t)
   &=\overline{\psi_1(p,n,t)} + R_1(p,t)\psi_1(p,n,t), \quad p \in\Sigma_1,
\end{align}
where $T(p,t)$, $R(p,t)$ (resp. $T_1(p,t)$, $R_1(p,t)$) are the right
(resp.\ left) transmission and reflection coefficients. 
They satisfy 
\be\label{SP8}\aligned
T(\ol p,t)&=\ol{T(p,t)},\quad  R(\ol p,t)=\ol{R(p,t)},\quad p\in\Sigma,\\ 
T_1(\ol p,t)&=\ol{T_1(p,t)},\quad  R_1(\ol p,t)=\ol{R_1(p,t)},\quad p\in\Sigma_1,\endaligned
\ee
and the identities
\be\label{iddd}
\frac{T(p,t)}{\overline{T(p,t)}} = R(p,t), \quad p \in\Sigma, \quad \frac{T_1(p,t)}{\overline{T_1(p,t)}} = R_1(p,t), \quad  p \in\Sigma_1.
\ee
If the coefficients of the Jacobi operator $\mathcal H(t)$ tend to their constant asymptotics with finite first moment (which is a more general situation than \eqref{decay}), then the transmission coefficients can be continued as meromorphic functions  on $\Pi_U$ with a simple pole at $p_0=(\la_0, +)$, and satisfy 
\be \label{TW}
T(p,t)=\frac{\sqrt{p^2-1}}{W(p,t)}, \quad T_1(p,t)=\frac{\sqrt{(p+c)^2 - 4d^2}}{W(p,t)}, \quad p\in\clos \Pi_U.
\ee
Moreover, $T(p,t)$ (resp.\ $T_1(p,t)$) is continuous in a vicinity of $\Gamma=\Sigma\cup\Sigma_1$ up to the boundary, excluding possibly the points $\pa \si (H_1)$ (resp.\ $\pa\si(H)$), where a discontinuity can 
appear due to the resonance. 
If $E\in\pa \si (H_1)$ (resp.\ $E\in\pa\si(H)$) is the resonant point then $T(p,t)=O((p-E)^{-1/2})$ (resp.\ $T_1(p,t)=O((p-E)^{-1/2})$), i.e., this transmission coefficient has  a simple pole at such a point.

Now we observe that under condition \eqref{decay} the reflection coefficients can be continued in the domain 
$\mathfrak D$. It is natural to continue them via \eqref{rsr} and \eqref{lsr},
\begin{align} \label{sr2}
\begin{split}
  \breve{\psi}(p,n,t)
  & = T(p,t) \psi_1(p,n,t) -R(p,t)\psi(p,n,t), \\  
 \breve{\psi_1}(p,n,t)
  & = T_1(p,t) \psi(p,n,t) -R_1(p,t)\psi_1(p,n,t), 
\end{split}
\qquad p \in\mathfrak D,
\end{align}
which is the same as to introduce them as usual via Wronskians (see \eqref{TW}, \eqref{cont66}),
$$R_1(p,t)=\frac{\langle \breve\psi_1,\,\psi\rangle(p,t)}{W(p,t)},\quad R(p,t)=-\frac{\langle \breve\psi,\,\psi_1\rangle(p,t)}{W(p,t)}.$$
In particular, \eqref{sr2} implies that both reflection coefficients also have simple poles at $p_0$.
Moreover, a pole for $R(p,t)$ (resp.\ $R_1(p,t)$) at the edge points of $\pa\si (H_1)$ (resp.\ $\pa\si(H)$) also appears in the resonance case. Thus, the following is valid:

\begin{lemma}\label{lemsc1}
Let $\mathfrak D$ be defined by \eqref{constabs}, \eqref{defel}. Then the functions
$$
T(p,t)\psi_1(p,n,t) - R(p,t)\psi(p,n,t)\ \mbox{ and } \ T_1(p,t)\psi(p,n,t) - R_1(p,t)\psi_1(p,n,t
)$$
are holomorphic in $\mathfrak D$ and continuous up to the boundary $\Gamma=\Sigma\cup\Sigma_1$.
\end{lemma}

Note that the time evolution of the scattering data preserves its form after analytical continuation. Set 
\[\phi(p) =\frac{1}{2}(z(p) - z^{-1}(p)),\quad \phi_1(p)= d( z_1^{-1}(p)-z_1(p) ),
\]
and denote
$$
\beta^{-1}(t)=\sum_{n\in\Z} 
(\psi(p_0,n,t))^2,\quad \beta_1^{-1}(t)=\sum_{n\in\Z} 
(\psi_1(p_0,n,t))^2,\quad \beta_1=\beta_1(0),\ \ \beta=\beta(0),
$$
and $T(p)=T(p,0)$, $T_1(p)=T_1(p,0)$, $R(p)=R(p,0)$, $R_1(p)=R_1(p,0)$, then we have
\be\label{R2}\begin{array}{lll}
T(p,t) =T(p)\E^{t(\phi(p)+\phi_1(p))}, & T_1(p,t)=T_1(p)\E^{t(\phi(p)+\phi_1(p))},  & p\in\Pi_U,\\[2mm]	 
R(p,t)=R(p)\E^{2t\phi(p)}, &
R_1(p,t)=R_1(p)\E^{2t\phi_1(p)}, & p \in \mathfrak D,\\[2mm]
\beta(t)=\beta\E^{2t\phi(p_0)}, &
\beta_1(t)=\beta_1 \E^{2t\phi_1(p_0)}.& 
\end{array}\ee

\subsection{Statement of the Riemann-Hilbert problem} 
 Let $m(p)=(m_1(p), m_2(p))$ be a vector-valued  function  on the Riemann surface $\M$, which has a jump on the contour $\Gamma$, oriented clockwise. 
 We will denote
 \[m_+(p)=\lim_{\zeta\in\Pi_U\to p\in\Gamma}m(\zeta),\quad m_-(p)=\lim_{\zeta\in\Pi_L\to p\in\Gamma}m(\zeta)\]
 at the same point $p\in\Gamma$. 
In general, for an oriented contour $\hat \Sigma$ on $\M$, and for a function $f(p)$ on this surface, the value $f_+(p)$ (resp.\ $f_-(p)$) will denote the nontangential limit of the vector function $f(\zeta)$ as $\zeta\to p \in \hat \Sigma$ from the positive (resp.\ negative) side of $\hat \Sigma$, where the positive side is the one which lies to
the left as one traverses the contour in the
direction of its orientation.

 We say that the vector-function $m$ satisfies 
\begin{itemize}\item {\it the symmetry condition} if
\be\label{symto}
m(p^*) = m(p)\sigma_1;\quad \sigma_1:= \sigI;
\ee
\item {\it the normalization condition} if there exists \[ \lim_{p \to \infty_\pm} m(p) = (m_1(\infty_\pm), m_2(\infty_\pm))\] and
\be\label{eq:normcond}
\ m_1(\infty_\pm)\cdot m_2(\infty_\pm) = 1, \quad m_1(\infty_\pm) > 0.
\ee
\end{itemize}
On $\Pi_U$ define two vector-valued functions $m(p)=m(p,n,t)$ and $m^1(p)=m^1(p,n,t)$:
\begin{align} \label{defm}
  \begin{split}
m(p) &=
\begin{pmatrix} T(p,t) \psi_1(p,n,t) z^n(p),  & \psi(p,n,t)  z^{-n}(p) \end{pmatrix},   \\
	m^1(p) &=
	\begin{pmatrix} T_1(p,t) \psi(p,n,t) z_1^{-n}(p),  & \psi_1(p,n,t)  z_1^{n}(p) \end{pmatrix}.
  \end{split}
\end{align}
 They are considered as functions of the variable $p$, and $n$ and $t$ are parameters.
\begin{lemma}{\rm (\cite{emtstp2})}\label{asypm}
The functions $m(p)$ and $m^1(p)$  have the following asymptotic behavior
as $p \to \infty_+$:
\begin{align}\label{asm}
\begin{split}
m(p) &= \Big(A(n,t)\Big(1 - \frac{B(n -1,t)}{\la}\Big), \frac{1}{A(n,t)}\Big(1 + \frac{B(n,t)}{\la}\Big)\Big)
+  O\Big(\frac{1}{\la^{2}}\Big),\\
m^1(p) &= \Big(A_1(n,t)\Big(1 - \frac{B_1(n +1,t)}{\la}\Big), \frac{1}{A_1(n,t)}\Big(1 + \frac{B_1(n,t)}{\la}\Big)\Big)
+  O\Big(\frac{1}{\la^{2}}\Big),
\end{split}
\end{align}
where 
  \begin{align}
    \begin{split} \label{AB}
  A(n,t) &= \prod_{j=n}^\infty 2a(j,t), \qquad
  B(n,t)= - \sum_{j=n+1}^\infty b(j,t), \\[-1mm]
  A_1(n,t) &= \prod_{j=- \infty}^{n-1} \frac{a(j,t)}{d}, \qquad
  B_1(n,t) = - \sum_{j=-\infty}^{n-1} (c + b(j,t)).
   \end{split}
  \end{align}
\end{lemma}
Extend the functions $m$ and $m_1$  to $\Pi_L$ by the symmetry condition,
$m(p^*) = m(p) \si_1,$ $m_1(p^*)=m_1(p)\si_1$. 
Evidently, this extension produces jumps along $\Gamma$.
To apply the nonlinear steepest descent method we have to describe the jumps along $\Gamma$ by
matrices depending on a large parameter $t$ and on a parameter $\xi=\frac{n}{t}$, which does not change much ({\it the slow variable}).
To this end, introduce the phase functions $\Phi(p)=\Phi(p,\xi)$  and $\Phi_1(p)=\Phi_1(p,\xi)$ on $\Pi_U$,
\be  \label{Phi3}
\Phi_1(p) = d\big(z_1^{-1}(p)-z_1(p)\big) - \xi\log z_1(p),\
\Phi(p) = \frac{1}{2}\big(z(p)-z^{-1}(p)\big) + \xi\log z(p),
\ee
and continue them as odd functions to $\Pi_L$
\be\label{oddphi}
\Phi(p^*)=-\Phi(p), \quad \Phi_1(p^*)=-\Phi_1(p).
\ee
This corresponds to the continuation $z(p^*)=z^{-1}(p)$ and $z_1(p^*)=z_1^{-1}(p)$, which is natural 
for the Joukovski transformation.
With this continuation, $z$ and $z_1$ are not holomorphic on $\M$; $z(p)$ has a jump on 
$\Sigma_1$ and $z_1(p)$ has a jump on $\Sigma$.
In particular,
\be \label{jumpz_1}
z_1(p)=z_1(\ol p)=z_1^{-1}(p^*)\in \R,  \quad p \in \Sigma,
\ee
and $z(p)$ is real-valued with the same type of jump on $\Sigma_1$.
Denote
\be\label{defchi}
 \chi(p) = -\lim_{p^\prime \in \Pi_U \to p \in \Sigma} T_1 (p^\prime,0)\, \ol {T(p^\prime,0)},
 \quad p \in \Gamma.
\ee
We observe from \eqref{TW} that
\be\label{propchi}
\chi(p)=\I |\chi(p)|, \quad p\in \Gamma_{u},
 \quad\chi(p)=-\I |\chi(p)|, \quad p\in \Gamma_{\ell},
 \ee
(cf. \eqref{siul}), and therefore $\chi(p) = - \chi({\ol p})$ for $p\in\Gamma$.

\begin{theorem}\label{thm:vecrhp}
Suppose that the initial data of the Cauchy problem \eqref{tl} satisfy \eqref{ini}--\eqref{cont12}.
Let $\{ R(p), p \in \Sigma;\  R_1(p),p\in\Sigma_1;\  \chi(p), \,p \in \Gamma;\  p_0=(\la_0,+), \,\beta_1,\,\beta_2 \}$ be
the  scattering data of $\mathcal H(0)$. Then the vector-valued functions defined 
in \eqref{defm}, \eqref{symto} solve the following  Riemann--Hilbert problems:

\begin{enumerate}[{\em I.}]\item The function $m(p)=(m_1(p), m_2(p))$ (resp.\ $m^1(p)=(m_1^1(p), m_2^1(p))$) is a meromorphic function on $\M \setminus \Gamma$ with a simple
pole at $p_0$ for $m_1(p)$ (resp.\ $m_1^1(p)$) and a simple pole at $p_0^*$ for $m_2(p)$ (resp.\ $m_2^1(p)$). 
It is continuous up to $\Gamma$ except at the points $(-c-2d, \pm)$ and $(-c+2d,\pm)$ 
(resp.\ $(1,\pm)$ and $(-1,\pm)$), 
where  $m_1(p)$ (resp.\ $m_1^1(p)$) admits a square root singularity  $O((p-E)^{-1/2})$ from the upper sheet, and 
$m_2(p)$ (resp.\ $m_2^1(p)$) from the lower sheet.
\item They satisfy the jump conditions $m_+(p)=m_-(p) v(p)$,  $m^1_+(p)=m^1_-(p) v_1(p)$, where
\be \label{eq:jumpcond}
v(p)= \begin{cases}
\begin{pmatrix}
0 & - \ol{R(p)} \E^{- 2 t \Phi(p)} \\[1mm]
R(p) \E^{2 t \Phi(p)} & 1
\end{pmatrix}, & p \in \Sigma,\\[4mm]
\begin{pmatrix}
\chi(p) \E^{t(\Phi_{+}(p)-\Phi_{-}(p))} & 1 \\[1mm]
1 & 0
\end{pmatrix}, & p \in \Sigma_1,
\end{cases}
\ee
\be \label{eq:jumpcond2}
v_1(p)=\begin{cases}
\begin{pmatrix}
0 & - \ol{R_1(p)} \E^{- 2 t \Phi_1(p)} \\[1mm]
R_1(p) \E^{2 t \Phi_1(p)} & 1
\end{pmatrix}, & p \in \Sigma_1,\\[6mm]
\begin{pmatrix}
\ol\chi(p) \E^{t(\Phi_{1,+}(p)-\Phi_{1,-}(p))} & 1 \\[1mm]
1& 0
\end{pmatrix}, & p \in \Sigma.
\end{cases}
\ee
\item They satisfy the pole conditions
\begin{align} \label{polecond}
 \res_{p_0} m(p) &= (Q m_2(p_0) , \ 0)
,\ \ \ 
\res_{p_0^*} m(p) = (0, Q m_1(p_0^*))\\ \label{polecond1}
 \res_{p_0} m^1(p) &= (Q_1 m_2^1(p_0) , \ 0)
,\ \ \ 
\res_{p_0^*} m^1(p) = (0, Q_1 m_1^1(p_0^*)),\end{align}
where
\[Q=Q(t)= \sqrt{p_0^2-1} \beta_2 \E^{2t\Phi(p_0)},\ \ 
Q_1=Q_1(t)= \sqrt{(p_0+c)^2 - 4d^2}
 \beta_1 \E^{2t\Phi_1(p_0)}
.\]
\item
They satisfy the symmetry and normalization conditions.
\end{enumerate}

\end{theorem}
\begin{proof} Use \eqref{SP1}, \eqref{rsr}, \eqref{lsr}, \eqref{SP8}, \eqref{iddd}, \eqref{R2}, \eqref{defchi}, 
\eqref{oddphi}, \eqref{jumpz_1}, \eqref{symto}, and \eqref{eq:normcond}.
Condition $I$ takes into account possible resonances, which produce poles of the transmission 
coefficients on the Riemann surface at the respective branch points. 
\end{proof}

\begin{lemma}\label{lemuniq} Each Riemann-Hilbert problem {\em I--IV} has a unique solution. 
\end{lemma}

\begin{proof}
Since the RH problems for $m$ and $m^1$ can be easily transformed into each other by a simple conjugation
	it suffices to study the uniqueness of $m$. Let $f=(f_1, f_2)$ and $g=(g_1, g_2)$ be two solutions satisfying I, 
	\eqref{eq:jumpcond}, \eqref{polecond}, and IV. For convenience we consider 
	them as in Section~\ref{sec:lr} as functions 
	on the Riemann surface $\hat \M$ of $\sqrt{\lambda^2 - 1}$. 
	The contour $\Sigma_1$ transforms to two contours: the interval $I_1=[-c-2d, -c + 2d]$ on the upper 
	sheet of $\hat \M$ oriented in positive direction, and $I_1^*$ on the lower sheet with negative orientation,
	with jump matrix 
	\[\hat v(p)=\begin{cases}
\begin{pmatrix}
1 & 0 \\
\chi(p)\E^{2t \Phi(p)}  & 1
\end{pmatrix}, & p \in I_1,\\[4mm]
\sigma_1( \hat v(p^*))^{-1}\sigma_1, & p \in I_1^*,\\
v(p),& p\in\Sigma.\end{cases}\] Let
	$$
	S(p)=\begin{pmatrix} f_1(p)& f_2(p)\\
	g_1(p) & g_2(p)\end{pmatrix}, \quad p\in\hat \M,
	$$ 
	then the scalar function $s(p)=\det S(p)$ has no jump since $\det v(p)=1$.  
	Moreover, $s(p)$ has no pole at the eigenvalue $p_0$ and is holomorphic on $\hat \M$ except at 
	four points $(-c-2d,\pm)$, $(-c+2d, \pm)$ (these are no longer branch points on $\hat\M$) 
	in the case of resonances, where
	$s(p) = O((p+c \pm 2d)^{-1/2})$.  
	Since $s(p)$ is bounded at $\infty_\pm$, then $s(p) \equiv const$ by Liouville's theorem. 
	The symmetry condition \eqref{symto} implies $s(p) + s(p^*)=0$, hence $2 s(1)=0$ and $s(p)\equiv 0$. 

	Therefore $f(p)=c(p)g(p)$, where $c(p)$ is a scalar function without jumps on $\hat\M$, and $\lim_{p\to\infty_\pm}c(p)=1$ by the normalization condition. 
	Hence to show uniqueness, it suffices to show that the associated vanishing problem, where the normalization condition \eqref{eq:normcond} is replaced by
	the condition that the first component of $m(\infty_+)$ vanishes, has only the trivial solution.
	To this end we introduce the meromorphic differential
	\[
	d\Omega(p) = \frac{\I\,d\la}{\pm\sqrt{\la^2-1}}, \quad p=(\la,\pm)\in\hat M,
	\]
	with simple poles at $\infty_\pm$. A brief inspection shows that
	$d\Omega$ is positive on $\Sigma$ and $\I^{-1} d\Omega$ is positive in $I_1$.
	
	Let $\hat m$ be a solution of this vanishing problem and
	let $\mathcal{C}$ be the closed contour from Fig.~\ref{fig:S2} oriented counterclockwise. Denote by $\hat m^\dagger$ the adjoint (transpose and complex conjugate)
	of a vector/matrix.
	Since there is no residue at $\infty_+$ we obtain
	$0= \int_{\mathcal{C}} \hat m(p)\hat m^\dagger(\overline{p^*})d\Omega(p)$,
	that is,
	\begin{align*}
	&2\pi\I \res_{p_0} \left(\hat m(p)\hat m^\dagger(\overline{p^*})\right)\frac{\I}{\sqrt{p_0^2 - 1}}=\\
	&= \int_{I_1} \hat m_+(p)\hat m^\dagger_+(\overline{p^*})d\Omega(p) - \int_{I_1} \hat m_-(p) \hat m^\dagger_-(\overline{p^*}) d\Omega(p) + \int_{\Sigma}\hat  m_+(p) \hat m^\dagger_-(p) d\Omega(p)\\
	&=  \int_{I_1} \big( \hat m_+(p) \sigma_1 \hat m^\dagger_-(p) - \hat m_-(p)\sigma_1  \hat m^\dagger_+(p) \big) d\Omega(p) + \int_{\Sigma} \hat m_+(p)  \hat m^\dagger_-(p) d\Omega(p).
	\end{align*}
	Using the pole condition \eqref{polecond}, the jump conditions
	\begin{align}\label{imp88}
	\hat m_{1,+}& - \hat m_{1,-}=\chi \E^{2t\re\Phi} \hat m_{2,-},\quad\hat m_{2,+}=\hat m_{2,-}, \quad \text{on } I_1,\\ \label{imp89}
	\hat m_{2,+}&=\hat m_{2,-}- \overline{R}\E^{-2t\Phi} \hat m_{1,-}, \quad
	\hat m_{1,+}=R\E^{2t\Phi}m_{2,-}, \quad \text{on } \Sigma,
	\end{align}
	together with $\chi(p)=\I|\chi(p)|,$  $p\in I_1$, and $\re \Phi(p)=0$, $p\in\Sigma$, imply
	\begin{align*}
	0&= 2\int_{I_1} |\chi(p)|\E^{2t\re\Phi(p)} |\hat m_{2,-}(p)|^2\I d\Omega(p) + \int_{\Sigma}|\hat m_{2,-}(p)|^2d\Omega(p)\\
	&\quad  + 4\pi \beta_2|\hat m_2(p_0)|^2\E^{2t\Phi(p_0)} + 2 \I  \im \int_{\Sigma} R(p) \E^{2t\Phi(p)} \hat m_{2,-}(p) \overline{\hat m_{1,-}(p)}d\Omega(p).
	\end{align*}
	Since the first three summands are positive and the last one is purely imaginary, this shows $\hat m_{2,-}(p)=0$ for $p\in I_1 \cup \Sigma$.
	By \eqref{imp88} $\hat m_{2,+}(p)=\hat m_{2,-}(p)=0$ for $p\in I_1$ and so $\hat m_1$ also has no jump along $I_1$.
	In particular, $\hat m$ is holomorphic in a neighborhood of $I_1$ and consequently vanishes on the upper sheet. By symmetry it also
	vanishes on the lower sheet which finally shows $\hat m(p)\equiv 0$ and establishes uniqueness.
\end{proof}

Our aim is to reduce these RH problems to model problems which can be solved explicitly.
To this end we record the following well-known result for easy reference.

\begin{lemma}[Conjugation] \label{lem:conjug}
Let $\tilde m$ be a solution of the RH problem $\tilde m_+(p)=\tilde m_-(p) \tilde v(p)$, $p \in\tilde\Sigma$, on a Riemann surface $\tilde \M$ which satisfies the symmetry and normalization conditions.
Let $\hat \Sigma$ be a contour on $\tilde\M$ with the same orientation as $\tilde\Sigma$ on the common part of these contours and suppose that $\hat\Sigma$ and $\tilde\Sigma$
contain with each point $p$ also $p^*$. Let $D$ be a matrix of the form
\[
D(p) = \begin{pmatrix} d(p)^{-1} & 0 \\ 0 & d(p) \end{pmatrix}=[d(p)]^{-\sigma_3},\quad \sigma_3=\begin{pmatrix} 1&0\\0&-1\end{pmatrix},
\]
where $d: \tilde\M\setminus\hat \Sigma\to\C$ is a sectionally analytic function with $d(p)\neq 0$ except for a
finite number of points on $\tilde\Sigma$. Set
\be\label{mD}\hat{m}(p) = \tilde m(p) D(p),\ee
then the jump matrix of the problem $\hat{m}_+=\hat{m}_- \hat{v}$ is
\[
\hat{v} =
\begin{cases}
  \begin{pmatrix} \tilde v_{11} &\tilde v_{12} d^{2} \\ \tilde v_{21} d^{-2}  & \tilde v_{22} \end{pmatrix}, &
  \quad p \in \tilde \Sigma \setminus (\tilde \Sigma\cap\hat\Sigma), \\[4mm]
  \begin{pmatrix} \tilde v_{11} d_+^{-1} d_-   &\tilde v_{12} d_+ d_- \\
 \tilde  v_{21} d_+^{-1} d_-^{-1}  &\tilde v_{22} d_-^{-1} d_+   \end{pmatrix}, &
  \quad p \in\tilde \Sigma\cap\hat\Sigma,
  \\[4mm]
  \begin{pmatrix} d_+^{-1} d_-   & 0 \\
  0  & d_-^{-1} d_+ \end{pmatrix}, &
  \quad p \in\hat\Sigma\setminus(\tilde \Sigma\cap\hat\Sigma).
\end{cases}
\]
If $d$ satisfies $d(p^*) = d(p)^{-1}$ for $p \in\tilde\M\setminus\hat \Sigma$,
then the transformation \eqref{mD} respects the symmetry condition \eqref{symto}.
\end{lemma}

In addition to this Lemma we will apply the technique of so called $g$-functions in a form proposed in \cite{KM}. In contradistinction to \cite{KM} we work on the Riemann surface, and
these $g$-functions are in fact Abel integrals on modified Riemann surfaces which are ``slightly truncated'' with respect to $\M$ and depend on the parameter $\xi$. These Abel integrals approximate the phase functions  at infinity up to an additive  constant, and
transform the jump matrices in a way that allows us to factorize them and to get asymptotically constant matrices on  contours. The respective RH problem with constant jump is called the model problem and will be solved explicitly for our case. In the next section we rigorously study the analytical properties of the $g$-function
which approximates the phase $\Phi$ in one of the domains, and then list analogous properties of the other $g$-functions in Section \ref{sec:el}.

\section{\texorpdfstring{$g$}{g}-function: existence and properties}\label{secg}

\subsection{Boundaries of regions} We start with properties of the phase functions which would be desirable to be ``inherited'' by the $g$-functions. In accordance with \eqref{Phi3} and \eqref{oddphi} we represent $\Phi(p)$ and $\Phi_1(p)$  via the integrals
\be \label{Phi}
\Phi(p) = - \int_{1}^{p}\frac{\la + \xi}{\sqrt{\la^2-1}}d\la,  \quad
\Phi_1(p) = \int_{-c+2d}^{p}\frac{\la + c+ \xi}{\sqrt{(\la + c)^2-4d^2}}d\la.
\ee
Evidently \eqref{oddphi} is valid. The function $\Phi(p)$ has a jump along the contour
$\Sigma_1$ and no jump on $\Sigma$, respectively, $\Phi_1$ has a jump along
$\Sigma$ and no jump on $\Sigma_1$. For $p\in\Pi_U$,
\be \label{jumpp}\aligned 
\Phi_\pm(p) &= \pm \I \pi \xi + \re \Phi(p), \quad \mbox{for } \pi(p) \in (-\infty, -1], \\
\Phi_{1,\pm}(p) &= \mp \I \pi \xi + \re \Phi_1(p), \quad \mbox{for } \pi(p) \in
(-\infty, -c-2d],\endaligned
\ee
with the natural symmetry on the lower sheet. The jumps of the phase functions along these intervals
are equal to $\frac{2\pi\I n}{t}$ up to a sign, which implies that $\E^{t(\Phi_{1,+}(p)- \Phi_{1,-}(p))}=1$ 
 along  contours on $\Pi_U$ and $\Pi_L$ with projection on $(-\infty, -c-2d]$, and $\E^{t(\Phi_{+}(p)- \Phi_{-}(p))}=1$ along two contours with projection on $(-\infty, -1]$.
The phase functions have the following asymptotic behavior as $p \to \infty_+$
\begin{align}\label{eqas7}
	\begin{split}
\Phi(p,\xi) &= - p - \xi\log p  - \xi \log2 + \frac{1}{2p} + O(p^{-2}),  \\
\Phi_1(p,\xi) &= p + \xi\log p + c  - \xi \log d + \frac{1}{p}(\xi c -2d^2) +  O(p^{-2}).
\end{split}
\end{align}
We observe that the graph $\re \Phi(p,\xi)=0$ (resp.\ $ \re \Phi_1(p,\xi)=0$) on $\clos (\Pi_U)$ consists of two curves. One of them is the contour $\Sigma$ (resp.\ $\Sigma_1$), and the other one crosses  the real axis  at the point $\eta$ (resp.\ $\eta_1$). If there is no confusion, we consider the real numbers $\eta$, $\eta_1$, and $\xi$ 
as points on $\Pi_U$ when necessary.  In particular, to evaluate the point $\eta$
observe that for $\xi>1$, the function $\Phi(p,\xi)$  maps the upper half-plane $\C^+\subset \Pi_U $ conformally to the domain
that lies below the polygon in the right picture of Fig.~\ref{fig:confmap}.
The line $\re \Phi(p,\xi)=0$ starts at $\eta<-\xi$ for which
\be
\Phi(\eta,\xi)=\Phi(-1,\xi)\label{eta}.
\ee
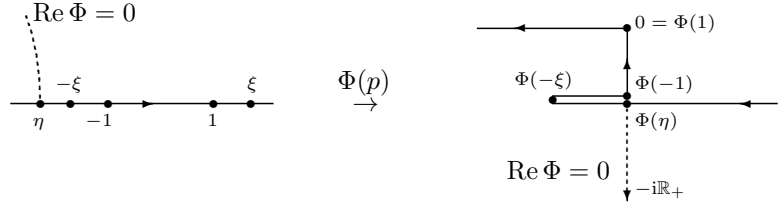
\begin{figure}[ht]
  \begin{picture}(3.5,3)

\linethickness{0.2mm}
  \put(0,1.2){\line(1,0){3.5}}
  \put(0.8,1.2){\circle*{0.09}}
  \put(0.6,1.4){\scriptsize{$-\xi$}}
   \put(2,1.3){}
  \put(1.8,1.2){\vector(1,0){0.1}}
  \put(1.3,1.2){\circle*{0.09}}
    \put(1,0.9){\scriptsize{$-1$}}
    \put(2.7,1.2){\circle*{0.09}}
       \put(2.62,0.9){\scriptsize{$1$}}
    \put(3.2,1.2){\circle*{0.09}}
      \put(3.15,1.4){\scriptsize{$\xi$}}

  \curvedashes{0.06,0.06}
  \curve(0.4,1.2, 0.35,1.85,  0.2,2.4)
   \put(0.3,2.3){$\re \Phi=0$}

  \put(0.4,1.2){\circle*{0.09}}
  \put(0.3,0.9){\scriptsize{$\eta$}}

  \end{picture}\quad
   \begin{picture}(2,3)
    \put(0.5,1.4){$\Phi(p)$}
    \put(0.7,1.1){$\to$}
    \end{picture}\quad
   \begin{picture}(4,3)

  \linethickness{0.2mm}
    \put(1,1.2){\line(1,0){3}}
    \put(3.6,1.2){\vector(-1,0){0.1}}
    \put(1,1.3){\line(1,0){1}}
     \put(1.01,1.25){\circle*{0.1}}
    \put(0.5,1.45){\scriptsize{$\Phi(-\xi)$}}
    \put(2,1.2){\circle*{0.09}}
      \put(2.1,0.9){\scriptsize{$\Phi(\eta)$}}
    \put(2,1.3){\line(0,1){0.9}}
    \put(2,1.7){\vector(0,1){0.1}}
    \put(2,1.3){\circle*{0.09}}
    \put(2.1,1.4){\scriptsize{$\Phi(-1)$}}
    \put(2,2.2){\circle*{0.09}}
    \put(2.1,2.2){\scriptsize{$0=\Phi(1)$}}
    \put(2,2.2){\line(-1,0){2}}
    \put(0.6,2.2){\vector(-1,0){0.1}}

    \curvedashes{0.06,0.06}
     \curve(2,1.2, 2,0)
         \put(2,0){\vector(0,-1){0.1}}
        \put(0.4,0.2){$\re \Phi=0$}
        \put(2.1,0){\scriptsize{$-\I\R_+$}}
    \end{picture}

\caption{Conformal map  $\Phi(p, \xi)$ for $\xi > 1$ and $p\in \Pi_U$, $\im p \geq 0$.}\label{fig:confmap}
\end{figure}

\noindent
Fig.~\ref{fig:confmap2} demonstrates that the curve $\re\Phi(p,\xi)=0$ starts at $\eta=-\xi$ when $\xi\in [-1,1]$.
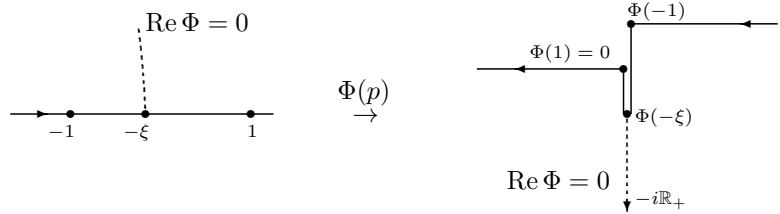
\begin{figure}[ht]
  \begin{picture}(3.5,3)

\linethickness{0.2mm}
  \put(0,1.2){\line(1,0){3.5}}
  \put(0.4,1.2){\vector(1,0){0.1}}
  \put(0.8,1.2){\circle*{0.09}}
  \put(0.5,0.9){\scriptsize{$-1$}}
   \put(2.5,1.3){ }
    \put(3.2,1.2){\circle*{0.09}}
      \put(3.15,0.9){\scriptsize{$1$}}

  \curvedashes{0.06,0.06}
  \curve(1.8,1.2, 1.76,1.85, 1.7,2.4)
   \put(1.8,2.3){$\small{\re \Phi=0}$}
  \put(1.8,1.2){\circle*{0.09}}
  \put(1.5,0.9){\scriptsize{$-\xi$}}

  \end{picture}\quad
   \begin{picture}(2,3)
    \put(0.5,1.4){$\Phi(p)$}
    \put(0.7,1.1){$\to$}
    \end{picture}\quad
    \begin{picture}(4,3)

    \linethickness{0.2mm}
       \put(2,1.2){\circle*{0.1}}
      \put(2.1,1.1){\scriptsize{$\Phi(-\xi)$}}
      \put(1.95,1.2){\line(0,1){0.6}}
      \put(2.05,1.2){\line(0,1){1.2}}
       \put(1.95,1.8){\circle*{0.09}}
       \put(2.05,2.4){\circle*{0.09}}
       \put(1.95,1.8){\line(-1,0){1.95}}
       \put(0.6,1.8){\vector(-1,0){0.1}}
       \put(0.7,1.95){\scriptsize{$\Phi(1)=0$}}

       \put(2.05,2.4){\line(1,0){1.95}}
       \put(2,2.5){\scriptsize{$\Phi(-1)$}}
        \put(3.6,2.4){\vector(-1,0){0.1}}

       \curvedashes{0.06,0.06}
         \curve(2,1.2, 2,0)
             \put(2,0){\vector(0,-1){0.1}}
            \put(0.4,0.2){$\small{\re \Phi=0}$}
            \put(2.1,0){\scriptsize{$-i\R_+$}}
      \end{picture}

\caption{Case $\xi \in (-1,1)$}\label{fig:confmap2}
\end{figure}

The signature table on $\M$ for $\Phi(p)$ in the case
$\eta\in I_1=[-c-2d, -c+2d]$ is given in Fig.~\ref{fig:signRePhi2}.
\begin{figure}
\begin{picture}(6,3.5)

\put(0,3.3){On $\Pi_U$:}
\put(1,0.55){$+$}
\put(1,2.3){$+$}
\put(4,2.3){$-$}
\put(4,0.55){$-$}

\linethickness{0.2mm}
\curve(1.6,0, 1.9,1.5, 1.6,3)
\put(1.9,1.5){\circle*{0.09}}
\put(1.97,1.65){\scriptsize $\eta$}
\put(1.8,2.5){$\re \Phi =0$}

\linethickness{0.2mm}
\put(1,1.5){\line(1, 0){1.5}}
\put(2.1,1.1){ $I_1$}

\linethickness{0.6mm}
\put(3.3,1.5){\line(1, 0){1.7}}
\put(4.7,1.1){ $\Sigma$}
\put(4,1.65){$\re \Phi =0$}

\end{picture}\quad
\begin{picture}(6,3.5)

\put(0,3.3){On $\Pi_L$:}
\put(1,0.55){$-$}
\put(1,2.3){$-$}
\put(4,2.3){$+$}
\put(4,0.55){$+$}

\linethickness{0.2mm}
\curve(1.6,0, 1.9,1.5, 1.6,3)
\put(1.9,1.5){\circle*{0.09}}
\put(1.97,1.65){\scriptsize $\eta^*$}
\put(1.8,2.5){$\re \Phi =0$}

\linethickness{0.2mm}
\put(1,1.5){\line(1, 0){1.5}}
\put(2.1,1.1){ $I_1^*$}

\linethickness{0.6mm}
\put(3.3,1.5){\line(1, 0){1.7}}
\put(4.7,1.1){ $\Sigma$}
\put(4,1.65){$\re \Phi =0$}

\end{picture}

\caption{Signature table of $\re \Phi(p)$ for $\eta \in I_1$}\label{fig:signRePhi2}
\end{figure}
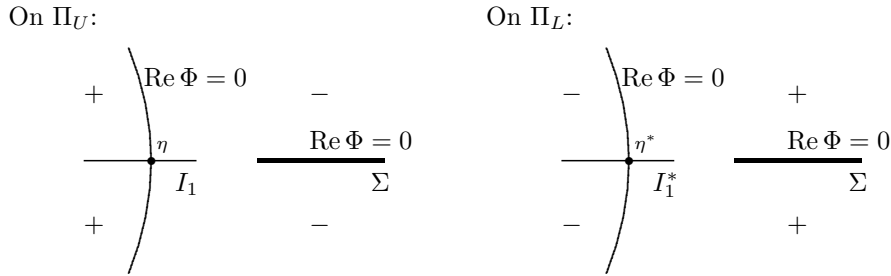
We observe that as the parameter $\xi$ decreases from $+\infty$ to $-\infty$, the point $\eta$ increases from $-\infty$ to $+\infty$ and the point $\eta_1$ decreases from $+\infty$ to $-\infty$. One can expect
from the signature table in Fig.~\ref{fig:signRePhi2} that for $\xi=\frac{n}{t}>\xi_{cr}$, where $\xi_{cr}$ corresponds to $\eta=-c-2d$, the asymptotical behavior of the solution
of  \eqref{tl}, \eqref{ini}--\eqref{cont12} will be close to the
coefficients of the right initial background operator $H$. Respectively, if $\xi<\xi_{cr,1}$, where $\xi_{cr,1}$ corresponds to $\eta_1=1$, the solution will be close to the coefficients of the left
background operator $H_1$ (see Section \ref{sec:lr}). According to \eqref{eta}, $\xi_{cr}$ is the solution of the equation $\Phi(-c-2d,\xi)=\Phi(-1,\xi)$.  From \eqref{Phi} we obtain
\be
\label{xicrit2}\xi_{cr}=\big((c+2d)^2 -1\big)^{1/2}\, \log^{-1}\big(c+2d +((c+2d)^2-1)^{1/2}\big).
\ee
Here the positive value of $\sqrt{\cdot}$ is used. In turn,
the point $\xi_{cr,1}$ is the solution of the equation $\Phi_1(1,\xi)=0$, that is
\be \label{xicrit1}
\xi_{cr,1}=\big((1+c)^2-4d^2\big)^{1/2} \Big(\log 2d - \log \big(1+c + ((1+c)^2-4d^2)^{1/2}\big)\Big)^{-1}.
\ee
We observe that $\xi_{cr,1}<-2d$ and $\xi_{cr}>1$, therefore $\xi_{cr,1} < \xi_{cr}$.
To determine the parameters which distinguish  four other regions of the $(n,t)$ half plane
where the solution has different types of finite-gap asymptotical behavior, we introduce 
the points $\nu_1, \nu_2\in(-c+2d,-1)$  such that
\be\label{mui}
\int_{-c+2d}^{-1}\frac{(\la - \nu_2)(\la +c - 2d)}{R^{1/2}(\la)}d\la = 0,\quad \int_{-c+2d}^{-1}\frac{(\la - \nu_1)(\la +1)}{R^{1/2}(\la)}d\la = 0.
\ee
Explicitely one obtains
\be
\nu_1 = \frac{(c-2d) \mathcal{I}_1 +  \mathcal{I}_2}{(c-2d) \mathcal{I}_0 +  \mathcal{I}_1}, \quad
\nu_2 = \frac{\mathcal{I}_1 +  \mathcal{I}_2}{\mathcal{I}_0 +  \mathcal{I}_1}, \quad
\mathcal{I}_\ell= \int_{-c+2d}^{-1}\frac{\la^\ell}{R^{1/2}(\la)}d\la,
\ee
where the integrals $\mathcal{I}_\ell$ can be explicitly evaluated in terms of Jacobi elliptic functions \cite{BrFr} yielding
\begin{align*}
\mathcal{I}_0 &= C K(k), \quad k=\sqrt{\frac{c^2-(2 d+1)^2}{c^2-(1-2 d)^2}}, \quad C=\frac{2}{\sqrt{c^2-(1-2 d)^2}}\\
\mathcal{I}_1 &= C \left( 4d\, \Pi\Big(\frac{1 - c + 2 d}{1 - c - 2 d},k\Big)-(c+2d) K(k) \right),\\
\mathcal{I}_2 &= \frac{C}{2} \left( 4 d K(k) + (1 - c - 2 d) (1 + c - 2 d) E(k) \right) - c \mathcal{I}_1.
\end{align*}
Set
\be\label{xicr}
\xi_{cr}^\prime=-\nu_2 - 2d,\quad \xi_{cr,1}^\prime=-\nu_1 - c +1.
\ee
Since  $\nu_1,\ \nu_2\in(-c+2d,-1)$, then $|\nu_2 - \nu_1|<-1+c-2d$. Therefore,
$
\xi_{cr,1}^\prime <\xi_{cr}^\prime.
$

To compute the critical value $\xi_{cr,0}$ which corresponds to the eigenvalue $\la_0$, introduce two functions $\mu_1(\xi) < \mu_2(\xi)$  uniquely defined by 
\eqref{condgap} for $\xi \in (\xi_{cr,1}^\prime, \xi_{cr}^\prime)$. 
Observe that for $\xi=\xi_{cr}^\prime$, we have $\mu_1=-c+2d$ and $\mu_2=\nu_2$. 
Respectively, for $\xi=\xi_{cr,1}^\prime$, $\mu_2=-1$ and $\mu_1=\nu_1$. Since 
$\la_0 \in (-c+2d,-1)$, the parameter $\xi_{cr,0}$ is defined by 
\be \label{cond14}
\int_{\la_0}^{-1}\frac{(\la - \mu_1)(\la - \mu_2)}{R^{1/2}(\la)}d\la = 0,
\ee
and therefore 
$\xi_{cr,1}^\prime< \xi_{cr,0} <\xi_{cr}^\prime$.
In fact, the following inequalities are valid
\be \label{ineqxi}
 \xi_{cr,1}<\xi_{cr,1}^\prime< \xi_{cr,0}<\xi_{cr}^\prime<\xi_{cr},
\ee
 where the parameters $\xi_{cr}$ are uniquely defined  by \eqref{xicrit2}--\eqref{cond14} and
 \eqref{condgap}.
These inequalities define three regions with different $g$-functions. In the region $\xi\in(\xi_{cr}^\prime, \,\xi_{cr})$ (resp.\ $\xi\in(\xi_{cr,1}, \,\xi_{cr,1}^\prime)$), a $g$-function will be a good
approximation for the phase function $\Phi$ (resp.\ $\Phi_1$), in the middle region a $g$-function will approximate both phase functions up to the sign. More precisely, in the right (resp.\ left) region we study the RH problem associated
with the right (resp.\ left) scattering data and in the middle we study both problems and compare solutions. The inequalities \eqref{ineqxi} can be verified directly, but we get them as a byproduct of existence of such $g$-functions.

\subsection{Definition and properties of the $g-$function for $\xi \in (\xi_{cr}^\prime,\, \xi_{cr})$} \label{gif} 
With the boundaries of the domains in place, we start by introducing the $g$-function for $ \xi_{cr}^\prime < \xi< \xi_{cr}$. 
Consider two real-valued functions, $\gamma(\xi) \in (-c-2d, -c +2d)$ and
$\mu(\xi) \in (\gamma(\xi), -1)$, such that the following two conditions are satisfied:
\be \label{eqas}
c+2d + \gamma(\xi) + 2 \mu(\xi)=-2\xi,
\ee
and
\be \label{zerocond}
\int_{\gamma(\xi)}^{-1}\frac{(\la - \mu(\xi))(\la - \gamma(\xi))}{R^{1/2}(\la, \gamma(\xi))}d\la = 0,
\ee
where
\be \label{Rxi}
R^{1/2}(\la,\gamma):= -\sqrt{(\la^2-1)(\la+c+2d)(\la-\gamma)}.
\ee
Evidently, we can always choose two points $\gamma\in (-c-2d, -c +2d)$ and $\mu(\gamma) \in (\gamma, -1)$ such
that \eqref{zerocond} holds true. Hence our aim is to show that in the given region, \eqref{eqas} can also be satisfied.

\begin{lemma} \label{abelint}
For any $\xi\in (\xi_{cr}^\prime, \xi_{cr})$ the following is valid:

\begin{enumerate}[(i)] 
 \item 
There exist points $\gamma(\xi)\in (-c-2d, -c +2d)$ and $\mu(\xi)\in (\gamma(\xi),-1)$ satisfying \eqref{eqas}--\eqref{zerocond}, they can be chosen uniquely.
\item
There exist points $\nu_1(\xi),  h(\xi)\in(\gamma(\xi), -1)$ and $\nu_2(\xi)\in\R$ such that
\be\label{cond37}
\frac{(\la - \mu(\xi))(\la-\gamma(\xi))}{R^{1/2}(\la,\gamma(\xi))}=\Omega(\la,\xi) +\xi\omega(\la,\xi),
\ee
where
\[
\Omega(\la,\xi)=\frac{(\la - \nu_1(\xi))(\la-\nu_2(\xi))}{R^{1/2}(\la,\gamma(\xi))},\quad
 \omega(\la,\xi)=\frac{\la - h(\xi)}{R^{1/2}(\la,\gamma(\xi))}
\]
with
\be\label{cond35} \text{\rm (a)}\ \ \int_{\gamma(\xi)}^{-1}\Omega(\la,\xi)d\la=0;\ \ \ \ \text{\rm (b)}\ \ \int_{\gamma(\xi)}^{-1}\omega(\la,\xi)d\la=0,\ee
and
 \be\label{cond333}
 \Omega(\la,\xi)=- 1 +O(\la^{-2}),\quad 
 \omega(\la,\xi)= - \frac{1}{\la} +O(\la^{-2}), \quad
 \mbox{as $\la\to\infty$}.
\ee
\item The following formula is valid
\be\label{deri}
\frac{\pa}{\pa\xi}\frac{(\la - \mu(\xi))(\la-\gamma(\xi))}{R^{1/2}(\la,\gamma(\xi))}=\omega(\la,\xi).
\ee
\item
The point $\gamma(\xi)$ moves continuously to the right from $\gamma(\xi_{cr})=-c-2d$ to 
$\gamma(\xi_{cr}^\prime)=-c+2d$ as $\xi$ decreases.
\end{enumerate}
\end{lemma}
This Lemma is proved in  Appendix A.

Now, let $\M(\xi)$ be the Riemann surface of the function \eqref{Rxi}, and denote by $\Pi_U(\xi)$ and $\Pi_L(\xi)$ its upper and lower sheets.
Set $I(\xi):=[\gamma(\xi),-1]$ and $I_2:=(-\infty, -c-2d]$, and consider these intervals as contours on  $\Pi_U(\xi)$ oriented in positive direction.
\begin{figure}[ht]
\begin{picture}(7,2.8)

\linethickness{0.2mm}
\put(-0.9,1.8){$\Pi_U$}
\put(1,2.1){\line(1, 0){2}}
\put(1,2){\line(1, 0){2}}
\put(1,2.05){\circle*{0.1}}
\put(0.5,1.75){\scriptsize $-c-2d$}
\put(3,2.05){\circle*{0.1}}
\put(2.8,1.75){\scriptsize $\gamma(\xi)$}

\put(3,2.05){\line(1, 0){1,5}}
\put(3.8,2.05){\vector(1,0){0.1}}
\put(3.6,2.25){$I(\xi)$}

\put(4.5,2.05){\circle*{0.1}}
\put(4.2,1.75){\scriptsize $-1$}
\put(4.5,2.1){\line(1, 0){2,3}}
\put(4.5,2){\line(1, 0){2,3}}
\put(6.8,2.05){\circle*{0.1}}
\put(6.7,1.75){\scriptsize $1$}

\put(-0.9,0.3){$\Pi_L$}
\put(1,0.6){\line(1, 0){2}}
\put(1,0.5){\line(1, 0){2}}
\put(1,0.55){\circle*{0.1}}
\put(3,0.55){\circle*{0.1}}

\put(3,0.55){\line(1, 0){1,5}}
\put(3.8,0.55){\vector(-1,0){0.1}}
\put(3.6,0.75){$I^*(\xi)$}

\put(4.5,0.55){\circle*{0.1}}
\put(4.5,0.6){\line(1, 0){2,3}}
\put(4.5,0.5){\line(1, 0){2,3}}
\put(6.8,0.55){\circle*{0.1}}

\linethickness{0.05mm}
\curve(-2,1.5, -1,2.5)
\curve(-2,1.5, 0.5,1.5)

\curve(-2,0, -1,1)
\curve(-2,0, 0.5,0)

\end{picture}
  \caption{The Riemann surface $\M(\xi)$ }\label{fig:cuts}
\end{figure}
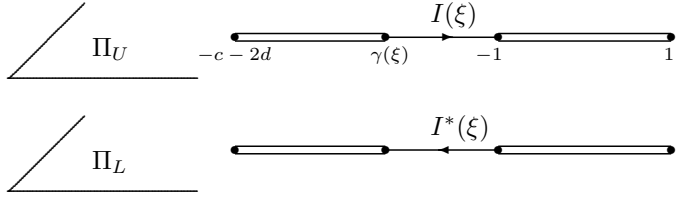

 Let $I^*(\xi)$ and $I_2^*$ be the respective contours on $\Pi_L(\xi)$ with negative direction.
Denote $ \mathbb D(\xi):=\M(\xi)\setminus \left(I(\xi)\cup I^*(\xi)\cup I_2\cup I_2^*\right) $ and 
for $p\in\Pi_U(\xi)\cap \mathbb D(\xi)$, introduce the function
\be\label{defgee}
 g(p):= g(p, \xi)=  \int_{1}^p
 \frac{(\la - \mu(\xi))(\la-\gamma(\xi))}{R^{1/2}(\la,\gamma(\xi))} d \la
\ee
and continue it as an odd function to the lower sheet,
\be \label{godd}
g(p^*)=-g(p).
\ee
Then $g$ is a singe-valued function on $\mathbb D(\xi)$. By \eqref{eqas}, it has the asymptotical behavior
\be\label{eqas5}
\Phi(p,\xi)- g(p,\xi)=K(\xi)  - \frac{2k(\xi)-1}{2p} + O(p^{-2})  \quad \mbox{ as $p\to\infty_+$},
\ee
where  $k(\xi)$ is the real-valued coefficient for the term of order $\frac{1}{p}$ in the expansion of $g(p, \xi)$ with respect to large $p\in\Pi_U(\xi)$, and
\be\label{Kxi}K(\xi)=\lim_{\la\to +\infty}\int_1^\la\left(\frac{(x-\mu(\xi))\sqrt{x-\gamma(\xi)}} {\sqrt{(x^2 - 1)(x+c+2d)}}-\frac{x+\xi}{\sqrt{x^2-1}}\right)dx
\ee
is a real constant.

Lemma~\ref{abelint}, (ii)--(iii), demonstrates an essential property of the $g$-function. 
It is an Abel integral on the Riemann surface $\M(\xi)$ which is represented as a linear 
combination of the normalized Abel integrals of the second and third kind, see Section~\ref{sec:model} for details. 
Both Abel integrals of the second and third kind depend on the parameter $\xi$, but \eqref{deri} shows that the 
derivative of the $g$-function with respect to $\xi$ can be expressed in terms of the Abel integral of the third 
kind only. Another important property of the $g$-function is depicted in Fig.~\ref{fig:Reg}.

\begin{figure}[ht]
\begin{picture}(9,3.2)

\linethickness{0.6mm}
\put(0,1.5){\line(1, 0){2.5}}
\put(0,1.5){\circle*{0.1}}
\put(-0.5,1.2){\scriptsize{$-c-2d$}}
\put(5.5,1.5){\line(1, 0){3.5}}
\put(5.5,1.5){\circle*{0.1}}
\put(5.2,1.2){\scriptsize{$-1$}}
\put(9,1.5){\circle*{0.1}}
\put(8.95,1.2){\scriptsize{$1$}}
\put(1.7,2.2){$+$}
\put(1.4,0.6){$+$}
\put(6,0.4){$-$}
\put(6,2.5){$-$}
\put(0,1.7){$\re g =0$}
\put(7.5,1.7){$\re g =0$}
\put(2.4,2.8){$\re g =0$}

\linethickness{0.2mm}
\curve(2.5,1.5, 3,2.2, 1.5,3)
\curve(2.5,1.5, 3,0.8, 1.5,0)
\put(2.5,1.5){\circle*{0.1}}
\put(2.1,1.2){\scriptsize{$\gamma(\xi)$}}

\curvedashes{0.05,0.05}
\curve(2.5,1.5, 5.5,1.5)
\put(3.8,1.5){\circle*{0.1}}
\put(3.3,1.2){\scriptsize $-c+2d$}
\put(4,1.7){$I(\xi)$}

\end{picture}
  \caption{Sign of $\re g$ on $\Pi_U(\xi)$}\label{fig:Reg}
\end{figure}
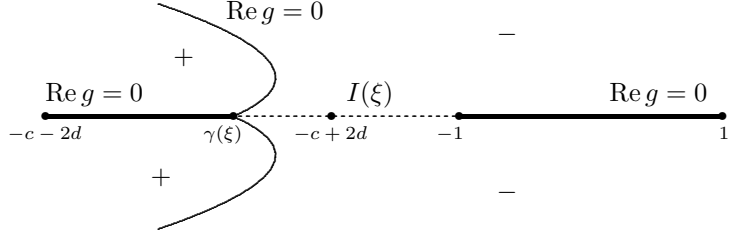
\noindent
We observe that the curve $\re g=0$ crosses the real axis namely at the branch point
$\gamma(\xi)$, which allows us to control the signature of the real part more accurately.
The signature table for $\re g$ 
has opposite signs on the lower sheet due to \eqref{godd}.

To describe the jumps of $ g(p, \xi)$ on $I(\xi)\cup I^*(\xi)\cup I_2\cup I_2^*$,
denote
\[
\int_{-c-2d}^{\gamma(\xi)} \frac{(\la - \mu(\xi))(\la-\gamma(\xi))}{R^{1/2}(\la+\I 0, \gamma(\xi))} d\la = \I B(\xi), \quad
\int_{-1}^{1} \frac{(\la - \mu(\xi))(\la-\gamma(\xi))}{R^{1/2}(\la+\I 0, \gamma(\xi))} d\la =\I B^\prime(\xi),
\]
where the integration is taken on $\Pi_U(\xi)$. Abbreviate $\Gamma(\xi):=\Sigma_1(\xi)\cup\Sigma$, where $\Sigma_1(\xi)$ is the contour on $M(\xi)$ along the interval $[-c-2d,\gamma(\xi)]$, oriented clockwise.
\begin{lemma}\label{lemjump} The function $g(p)$ satisfies the following properties:
\begin{align} \label{appen2} 
& g(p, \xi)  =- g(\ol p, \xi)\in\I\R, \quad p\in \Gamma(\xi), \\ \label{ju}
& \E^{t(g_+(p,\xi)- g_-(p,\xi))} =\E^{2\I t B(\xi)}, \quad  p\in I(\xi)\cup I^*(\xi), \\ \label{jumpge}
& \E^{t(g_+(p,\xi)- g_-(p,\xi))} =1, \quad p\in I_2\cup I_2^*.
\end{align}
\end{lemma}
\begin{proof} Property \eqref{appen2} follows from \eqref{zerocond}. Moreover, it is evident that
\begin{align}\label{jumpI3}
 g_+(p,\xi)-g_-(p,\xi)& =-2\I B^\prime(\xi),\quad \mbox{for $p\in I(\xi)\cup I^*(\xi)$}, \\ \label{jumpI4}
 g_+(p,\xi)- g_-(p,\xi)& =-2\I (B +B^\prime)(\xi),\quad \mbox{for $p\in I_2\cup I_2^*$}.
\end{align}
 Let $\mathcal{C}_\rho$ be a circle with radius $\rho$ and clockwise orientation enclosing the interval $[-c-2d,1]$ on the upper sheet.
By \eqref{zerocond},
\[
P(\xi):=\oint_{\mathcal{C}_\rho}\frac{(\la - \mu(\xi))(\la-\gamma(\xi))}{R^{1/2}(\la, \gamma(\xi))} d\la=2\I(B(\xi) + B^\prime(\xi)).
\]
On the other hand,
\[
P(\xi)=2\pi\I\res_{\infty_+}\frac{(\la - \mu(\xi))(\la-\gamma(\xi))}
 {R^{1/2}(\la, \gamma(\xi))} = -2 \pi \I \xi,
\]
due to \eqref{eqas7} and \eqref{eqas}, which implies \eqref{eqas5}.
Hence
\be\label{mainb}B(\xi)+B^\prime(\xi)=-\pi\xi,
\ee
which justifies \eqref{jumpI4}. Since $\xi=\frac{n}{t}$, we obtain $tP(\xi)=2\pi\I n$ and thus \eqref{jumpge}.
We also replace the jump \eqref{jumpI3} by the jump
\[
 g_+(p,\xi)- g_-(p,\xi)=2\I B(\xi)+2\pi\I\xi, \quad p \in I(\xi)\cup I^*(\xi),
\]  from which \eqref{ju} follows.
\end{proof}

\section{Reduction to the model problem for \texorpdfstring{$\xi \in (\xi_{cr}^\prime, \xi_{cr})$}{xi,prime,cr<xi<xi,cr}}
\label{sec:red}

In this section we perform four basic conjugation/deformation steps which allow us to transform the initial RH problem to an equivalent RH problem with a jump matrix close to a constant matrix for large $t$, except for neighborhoods of the points $-c-2d$ and $\gamma(\xi)$. Up to a natural symmetry these steps will be the same in all domains under consideration, and all of them are invertible. Moreover, each step preserves the symmetry condition 
\eqref{symto} and the normalization condition \eqref{eq:normcond}.

Let  $m(p)$ be the solution of the RH problem described in Theorem \ref{thm:vecrhp},
considered for the values $\xi_{cr}^\prime < \xi < \xi_{cr}$. 

\noindent {\it Step 1}. Let  $\M(\xi)$ be the Riemann surface introduced  in Subsection~\ref{gif} and consider 
$m(p)$ as a function on $\M(\xi)$.  Using the symmetry property \eqref{symto} we  rewrite the initial  RH problem as a problem on $\M(\xi)$ with complementary jumps along the contours $[\gamma, -c+2d]=I_3(\xi)\subset \Pi_U(\xi)$ oriented from left to right, and $I_3^*(\xi)\subset \Pi_L(\xi)$, oriented from right to left. The function
\[
 \chi(p) = -\lim_{p^\prime \in \Pi_U \to p \in \Sigma_{1,u}} T_1 (p^\prime,0)\, \ol {T(p^\prime,0)}, 
\] 
is considered as a function on $I_3(\xi)$. Respectively,
 $\chi(p)=-\chi(p^*)$ for $p\in I_3^*(\xi)$. Thus we get an equivalent holomorphic RH problem 
 on $\M(\xi)$ for $m^{(1)}(p)=m(p)$: 
to find a meromorphic function on   $\M(\xi)\setminus\left( \Sigma_1(\xi)\cup I_3(\xi)\cup I_3^*(\xi)\cup \Sigma_2\right)$, satisfying conditions \eqref{polecond}, \eqref{symto}, \eqref{eq:normcond}, and the jump condition 
\noindent $m^{(1)}_+(p)=m^{(1)}_-(p)v^{(1)}(p)$, where
\be\label{jumpdva}
v^{(1)}(p) = \begin{cases}
\begin{pmatrix}
\chi(p)\E^{t(\Phi_+(p) - \Phi_-(p)) } & 1 \\
1 & 0
\end{pmatrix},& p \in \Sigma_1(\xi),\\[4mm]
\begin{pmatrix}
1 & 0 \\
\chi(p)\E^{2t \Phi(p)}  & 1
\end{pmatrix}, & p \in I_3(\xi),\\[4mm]
\sigma_1( v^{(1)}(p^*))^{-1}\sigma_1, & p \in I_3^*(\xi),\\
v(p), & p \in \Sigma.
\end{cases}
\ee

\noindent {\it Step 2}. Factorize the jump matrix $v^{(1)}(p)$ on $\Sigma$ using Schur complements,
\[
v^{(1)}(p)=\begin{pmatrix}
1 & - \ol{R(p)}\E^{-2t \Phi(p)} \\
0&1
\end{pmatrix}
\begin{pmatrix} 1&0\\
R(p)\E^{2t\Phi(p)} & 1
\end{pmatrix}.
\]
\begin{figure}[ht]
\begin{picture}(9,2.6)

\linethickness{0.2mm}
\put(0,1.5){\line(1, 0){2}}
\put(0,1.45){\circle*{0.1}}
\put(-0.5,1.15){\scriptsize $-c-2d$}
\put(0,1.4){\line(1, 0){2}}
\put(2,1.45){\circle*{0.1}}
\put(1.6,1.15){\scriptsize $\gamma(\xi)$}
\put(0.6,1.7){$\Sigma_1(\xi)$}
\put(1.1,1.5){\vector(1,0){0.1}}
\put(0.9,1.4){\vector(-1,0){0.1}}

\put(2,1.45){\line(1, 0){3.5}}
\put(4,1.45){\circle*{0.1}}
\put(3.6,1.15){\scriptsize $-c+2d$}
\put(2.2,1.7){$I_4(\xi)$}
\put(2.5,1.45){\vector(1,0){0.1}}

\put(3,1.45){\circle*{0.1}}
\put(2.9,1.15){\scriptsize $y$}

\put(5.5,1.5){\line(1, 0){3}}
\put(5.5,1.4){\line(1, 0){3}}
\put(5.5,1.45){\circle*{0.1}}
\put(8.5,1.45){\circle*{0.1}}
\put(5.3,1.15){\scriptsize $-1$}
\put(8.45,1.15){\scriptsize $1$}
\put(6.7,1.7){$\Sigma(\xi)$}
\put(7.05,1.5){\vector(1,0){0.1}}
\put(6.85,1.4){\vector(-1,0){0.1}}

\linethickness{0.1mm}
\curve(3,1.5, 6.5,2.4, 9,1.45, 6.5,0.5, 3,1.4)
\put(6.5,2.4){\vector(1,0){0.1}}
\put(6.5,0.5){\vector(-1,0){0.1}}
\put(6.3,0.8){$\Omega(\xi)$}
\put(8,2.4){$\mathcal{C}(\xi)$}

\end{picture}
  \caption{Contour deformation for $m^{(1)}(p)$ on $\Pi_U$}\label{fig:m1}
\end{figure}
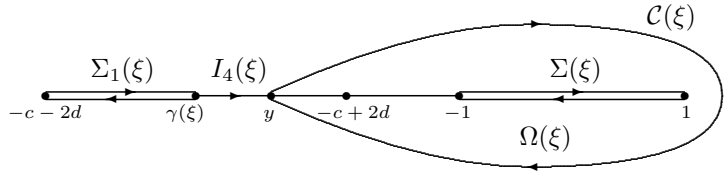

Let $y=y(\xi)$, $\gamma<y<-c+2d$, be the midpoint of $I_3(\xi)$ and let $\Omega(\xi)$
be a domain on the upper sheet as in Fig.~\ref{fig:m1}, inside the domain $\mathfrak D$ given by \eqref{defel}.
Recall that the reflection coefficient can be continued analytically inside $\mathfrak D$, and therefore to $\Omega(\xi)$, and has a pole at $p_0$.
We define $R(p)=\ol {R(p^*)}$ for $p\in\Omega^*(\xi)$. Set 
\be \label{m3}
m^{(2)}(p) = m^{(1)}(p) \begin{cases}
\begin{pmatrix} 1&0\\
-R(p) \E^{2t\Phi(p)} & 1
\end{pmatrix}, & p \in \Omega(\xi),\\[4mm]
\begin{pmatrix}
1 & - R(p^*) \E^{-2t \Phi(p)} \\
0&1
\end{pmatrix}, & p \in \Omega^*(\xi) ,\\[4mm]
\id, & \text{else}.
\end{cases}
\ee
This vector has no  jump along $\Sigma$. Instead it has jumps along the contours $\mathcal C(\xi)$ and 
$\mathcal C^*(\xi)$, which are oriented clockwise. Taking into account that $R(p,t)z^{2n}=R(p)\E^{2 t \Phi(p)}$, and \eqref{defm}, Lemma~\ref{lemsc1}, we conclude that the  components of  $m^{(2)}(p)$ have no poles at $p_0$ and $p_0^*$. Moreover, in the case of a resonance at  $\pi(p)=-c+2d$, the first component  $m^{(2)}_1(p)$ has no singularity at $(-c+2d, +)$ since by \eqref{defm} and Lemma \ref{lemsc1} we have
\[
\lim_{p\to -c+2d} m^{(2)}_1(p) = \lim_{p\to -c+2d} \big(m_1(p) - R(p)\E^{2t\Phi(p)}m_2(p)\big)=\mbox{const}<\infty.
\]
The same is true for the second component of  $m^{(2)}(p)$  at $(-c+2d,-)$.
Thus, $m^{(2)}(p)$ is the unique solution of the following RH problem: to find a holomorphic function 
on $\M(\xi)\setminus(\Sigma_1(\xi)\cup I_3(\xi)\cup I_3^*(\xi)\cup  \mathcal C(\xi)\cup \mathcal C^*(\xi))$, continuous up to the boundary (except possibly at $-c- 2d$,
where poles are admissible for one of the components), which satisfies the jump $m^{(2)}_+(p)=m^{(2)}_-(p)v^{(2)}(p)$ with
\be\label{jumpodin}
v^{(2)}(p) = \begin{cases}
\begin{pmatrix}
\chi(p)\E^{t(\Phi_+(p) - \Phi_-(p)) } & 1 \\
1 & 0
\end{pmatrix},& p \in \Sigma_1(\xi),\\[4mm]
\begin{pmatrix}
1 & 0 \\
\chi(p)\E^{2t \Phi(p)}  & 1
\end{pmatrix}, & p \in I_4(\xi),\\[4mm]
\begin{pmatrix}
1 & \chi(p)\E^{-2t \Phi(p)}  \\
0 & 1
\end{pmatrix}, & p \in I_4^*(\xi),\\[4mm]
\begin{pmatrix}
1 & 0 \\
-R(p) \E^{2 t \Phi(p)} & 1
\end{pmatrix}, & p \in \mathcal{C}(\xi),\\[4mm]
 \begin{pmatrix}
1 & -R(p^*) \E^{-2 t \Phi(p)} \\
0 & 1
\end{pmatrix}, & p \in \mathcal{C}^*(\xi),
\end{cases}
\ee
and standard normalization and symmetry conditions.

Here $I_4(\xi):=[\gamma(\xi), y(\xi)]\subset I_3(\xi)$ with the same orientation as $I_3(\xi)$. Orientation on  
$I_4^*(\xi)$ is preserved from $I_3^*(\xi)$. Note that the jump matrix on 
$I_3(\xi)\setminus I_4(\xi)$ (resp.\ $I_3^*(\xi)\setminus I_4^*(\xi)$) is equal to $\id$, because the lower (resp.\ upper) element of the jump matrix vanishes due to the identity (cf. \cite{egkt})
\be\label{pruff}
R_1(p) - R(p) + \chi(p)=0.
\ee 

\noindent {\it Step 3}. 
Denote
\be \label{defdg}
d(p) =\E^{t(\Phi(p) - g(p))},
\ee
where $g(p)=g(p,\xi)$ is defined by  \eqref{defgee}, and set
\be \label{m1}
m^{(3)}(p)=m^{(2)}(p)[d(p)]^{-\sigma_3}.
\ee
Lemma~\ref{lem:conjug} is applicable for this transformation by \eqref{godd}, \eqref{oddphi}.  Applying 
Lemma~\ref{lemjump} we obtain that the vector function $m^{(3)}(p)$ solves the following RH problem on $\M(\xi)$: $m^{(3)}_+(p)=m^{(3)}_-(p)v^{(3)}(p)$, where
\be\label{jumptre}
v^{(3)}(p) = \begin{cases}
\begin{pmatrix}
\chi(p) & \E^{-2tg(p)} \\
\E^{2 t g(p)} & 0
\end{pmatrix},& p \in \Sigma_1(\xi),\\[4mm]
\begin{pmatrix}\E^{2\I t B} & 0 \\
 0  & \E^{-2\I t B}
\end{pmatrix}, & p \in (I(\xi)\cup I^*(\xi))\setminus (I_4(\xi)\cup I_4^*(\xi)),\\[4mm]
\begin{pmatrix}
\E^{2\I t B} & 0 \\
\chi(p)\E^{2t \re g(p)}  & \E^{-2\I t B}
\end{pmatrix}, & p \in I_4(\xi),\\[4mm]
\begin{pmatrix}
\E^{2\I t B} & \chi(p)\E^{-2t \re g(p)}  \\
0 & \E^{-2 \I t B}
\end{pmatrix}, & p \in I_4^*(\xi), \\[4mm]
\begin{pmatrix}
1 & 0 \\
-R(p) \E^{2 t g(p)} & 1
\end{pmatrix}, & p \in \mathcal{C}(\xi),\\[4mm]
 \begin{pmatrix}
1 & -R(p^*) \E^{-2 t g(p)} \\
0 & 1
\end{pmatrix}, & p \in \mathcal{C}^*(\xi).
\end{cases}
\ee
Note that $m(\infty_\pm)=m^{(2)}(\infty_\pm)$, but the conjugation of Step 3 changes this value 
due to \eqref{eqas5},
\be \label{normm1}
m^{(3)}(\infty_+)=m(\infty_+)
 \begin{pmatrix}
\E^{-tK(\xi)} & 0 \\
0 & \E^{tK(\xi)}
\end{pmatrix}, \quad
m^{(3)}(\infty_-)=m^{(3)}(\infty_+)
 \sigma_1.
\ee
Note that $m^{(3)}$ still obeys \eqref{eq:normcond} and that $K(\xi)$ is defined by \eqref{Kxi}.
Moreover, the symmetry condition is preserved as well.

{\it Step 4}. Our next conjugation step deals with the factorization of the jump matrix on $\Sigma_1(\xi)$.
Consider the following scalar conjugation problem:

{\it  Find a bounded holomorphic function $F(p)$ on $\M(\xi)\setminus(\Sigma_1(\xi)\cup I(\xi)\cup I^*(\xi))$ with $F(p^*)=F^{-1}(p)$, real-valued at $\infty_\pm$, and
satisfying the jump conditions}
\be \label{ff}
F_+(p)=F_-(p) \begin{cases} |\chi(p)|, & p \in \Sigma_1(\xi), \\[1mm]
\E^{\I \Delta(\xi)}, & p\in I(\xi)\cup I^*(\xi).
 \end{cases}
\ee
 The value of the real constant $\Delta(\xi)$ will be specified later in \eqref{Deltaj}. As is shown in Lemma~\ref{omega3} below
 this problem is uniquely solvable.
 Given the solution of \eqref{ff} and taking into account that
$g$ has no jump on $\Sigma_1(\xi)$ and $\chi(p)=-\I |\chi(p)|$ for $p \in \Sigma_{1,\ell}(\xi)$ by \eqref{propchi},
we factorize the conjugation matrix $v^{(3)}(p)$ on $\Sigma_1(\xi)$ according to
\[
v^{(3)}(p) = \begin{cases}
 \begin{pmatrix}
F_-^{-1} & 0 \\
\frac{F_-^{-1}}{\chi} \E^{2 t g(p)} & F_-
\end{pmatrix}
\begin{pmatrix}
\I & 0 \\
0 & \I
\end{pmatrix}
 \begin{pmatrix}
F_+ & \frac{F_+}{\chi}\E^{-2 t g(p)} \\
0 & F_+^{-1}
\end{pmatrix}, & p \in \Sigma_{1,u}(\xi),  \\[4mm]
 \begin{pmatrix}
F_-^{-1} & 0 \\
\frac{F_-^{-1}}{\chi} \E^{2 t g(p)} & F_-
\end{pmatrix}
\begin{pmatrix}
-\I & 0 \\
0 & -\I
\end{pmatrix}
 \begin{pmatrix}
F_+ & \frac{F_+}{\chi}\E^{-2 t g(p)} \\
0 & F_+^{-1}
\end{pmatrix}, & p \in \Sigma_{1,\ell}(\xi).
\end{cases}
\]
Next introduce lens-shaped domains $\Omega_1(\xi)$ and $\Omega_1^*(\xi)$ around the contour
$\Sigma_1(\xi)$ as depicted in Fig.~\ref{fig:S1} with $\Omega_1(\xi)\subset \mathfrak D$
\begin{figure}[ht]
\begin{picture}(4,2)

\linethickness{0.2mm}
\put(0.5,1.5){\line(1, 0){3.5}}
\put(0.5,1.4){\line(1, 0){3.5}}
\put(2.2,1.5){\vector(1,0){0.1}}
\put(0.5,1.45){\circle*{0.1}}
\put(4,1.45){\circle*{0.1}}
\put(2,1.7){$\Sigma_{1}(\xi)$}
\put(1.4,0.9){$\Omega_1(\xi)$}
\put(3.3,2){$\mathcal{C}_1(\xi)$}

\linethickness{0.1mm}
\curve(4,1.4, 1.5,0.6, 0,1.45, 1.5,2.3, 4,1.5)
\put(1.5,2.3){\vector(1,0){0.1}}
\put(1.5,0.6){\vector(-1,0){0.1}}
\end{picture}\hspace{1.5cm}
\begin{picture}(4,3)

\linethickness{0.2mm}
\curvedashes{0.06,0.06}
\curve(0.5,1.5, 4,1.5)
\curve(0.5,1.4, 4,1.4)
\put(2.2,1.5){\vector(1,0){0.1}}
\put(0.5,1.45){\circle*{0.1}}
\put(4,1.45){\circle*{0.1}}
\put(2,1.7){$\Sigma_{1}(\xi)$}
\put(1.4,0.9){$\Omega_1^*(\xi)$}
\put(3.3,2){$\mathcal{C}_1^*(\xi)$}

\linethickness{0.1mm}
\curvedashes{0.06,0.06}
\curve(4,1.4, 1.5,0.6, 0,1.45, 1.5,2.3, 4,1.5)
\put(1.5,2.3){\vector(1,0){0.1}}
\put(1.5,0.6){\vector(-1,0){0.1}}
\end{picture}
  \caption{The lens contour near $\Sigma_1(\xi)$. Views from the upper and lower sheet.
  Dotted curves lie on the lower sheet.}\label{fig:S1}
\end{figure}
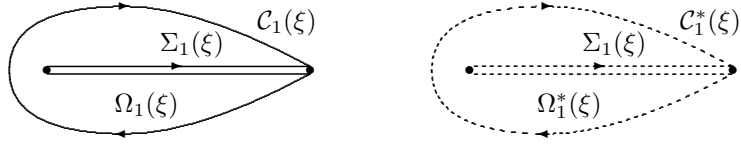
and transform the vector $m^{(3)}(p)$ as follows:
\be\label{em4}
m^{(4)}(p) =m^{(3)}(p)\begin{cases}
 \begin{pmatrix} F^{-1}(p) &  - \frac{F(p)}{V(p)}\E^{-2 t g(p)}\\
0 & F(p)
\end{pmatrix}, & p \in \Omega_1(\xi),\\[4mm]
\begin{pmatrix} F^{-1}(p) & 0 \\
\frac{F^{-1}(p)}{V(p)}\E^{2 t g(p)} & F(p)
\end{pmatrix}, & p \in \Omega_1^*(\xi),\\[4mm]
\begin{pmatrix} F^{-1}(p) & 0 \\
0 & F(p)
\end{pmatrix}, & \text{else},
\end{cases}
\ee
where $V(p)$ is an analytical continuation of $\chi(p)=-\ol{T(p)}T_1(p)$ to $\Omega_1(\xi)$ and $V(p^*)=-V(p)$ for $p\in\Omega_1^*(\xi)$. This transformation preserves the symmetry property for  $m^{(4)}$ due to $F(p^*)=F^{-1}(p)$. Since $F(\infty_+)\in\R_+$ (see Lemma \ref{omega3} below), the normalization property is also preserved for 
$m^{(4)}$. Recall that the only singularity at an edge point of the current jump contours which can happen 
for $m^{(3)}(p)$, is the singularity at $-c-2d$, where $m^{(3)}_1(p)=C(p+c+2d)^{-1/2}(1+o(1))$, $p\in\Pi_U$ in the resonant case. But in this case $\chi(p)=C(p+c+2d)^{-1/2}(1+o(1))$, $p\in\Sigma_1$, and therefore 
$F(p)=C (p+c+2d)^{-1/4}(1+o(1))$ for $p\in \Pi_U$ and $F(p)=C
(p+c+2d)^{1/4}(1+o(1))$ for $p\in\Pi_L$ (cf.\ \cite{mush}). Here $C$ denotes arbitrary non-vanishing constants. 
Thus, in a vicinity of $-c-2d$ we have by \eqref{em4} that
$m^{(4)}_1(p)=m^{(3)}_1(p)F^{-1}(p)=C((p+c+2d)^{-1/4})(1 + o(1))$  for  $p\in \Pi_U$  in the resonant case. 
In the nonresonant case $\chi(p)= C(p+c+2d)^{1/2}(1+o(1))$, and $F^{-1}(p)= C
(p+c+2d)^{-1/4}(1+o(1))$ for $p\in \Pi_U$. Consequently $m^{(4)}_1(p)=C((p+c+2d)^{-1/4})$, $p\in \Pi_U$,  in the nonresonant case too.
Denote  
\[
\Gamma(\xi):=\Sigma_1(\xi)\cup\Sigma\cup I(\xi)\cup I^*(\xi)\cup\mathcal C_1(\xi)\cup\mathcal C_1^*(\xi)\cup \mathcal C(\xi)\cup \mathcal C^*(\xi),
\]
where the orientation is chosen as before. We proved the following

\begin{theorem}\label{th:em4} For any $\xi\in (\xi_{cr}^\prime, \xi_{cr})$, the initial RH problem formulated in Theorem~\ref{thm:vecrhp} is equivalent to the following RH problem:
To find a  holomorphic vector function $m^{(4)}(p)$ in the domain $\M(\xi)\setminus \Gamma(\xi)$, with both components continuous up to the boundary 
except at the point $-c-2d$, where 
\be\label{resonance}
m^{(4)}_1(p)=m^{(4)}_2(p^*)=C 
(p+c+2d)^{-1/4}(1+o(1)),\quad  \mbox{for } p\in \Pi_U, \quad C\neq 0.
\ee
This vector function satisfies the symmetry and normalization conditions from Theorem~\ref{thm:vecrhp} 
and the jump condition $m^{(4)}_+(p)=m^{(4)}_-(p)v^{(4)}(p)$ with
\[
v^{(4)}(p) = \begin{cases}
\quad \I \id, & p \in \Sigma_{1,u}(\xi),\\
 \ -\I \id, & p \in \Sigma_{1,\ell}(\xi),\\
\begin{pmatrix}
\E^{2\I tB+\I\Delta} & 0 \\
0 & \E^{-2\I tB-\I\Delta}
\end{pmatrix}, & \hspace{-2.5cm}p \in (I(\xi)\cup I^*(\xi))\setminus (I_4(\xi)\cup I_4^*(\xi)),\\[4mm]
\begin{pmatrix}
\E^{2\I t B+\I\Delta} & 0 \\
\chi(p)F_+^{-1}(p)F_-^{-1}(p)\E^{2t \re g(p)}  & \E^{-2\I t B-\I\Delta}
\end{pmatrix}, & p \in I_4(\xi),\\[4mm]
\begin{pmatrix}
\E^{2\I t B+\I\Delta} & \chi(p)F_+(p)F_-(p)\E^{-2t \re g(p)}  \\
0 & \E^{-2 \I t B-\I\Delta}
\end{pmatrix}, & p \in I_4^*(\xi),\\[4mm]
 \begin{pmatrix}
1 & F^2(p)V^{-1}(p)\E^{-2 t g(p)} \\
0 & 1
\end{pmatrix}, & p \in \mathcal{C}_1(\xi),\\[4mm] 
\begin{pmatrix}
1 & 0 \\
- F^{-2}(p)V^{-1}(p)\E^{2 t g(p)} & 1
\end{pmatrix}, & p \in \mathcal{C}_1^*(\xi),\\[4mm]
\begin{pmatrix}
1 & 0 \\
-F^{-2}(p)R(p) \E^{2 t g(p)} & 1
\end{pmatrix}, & p \in \mathcal{C}(\xi),\\[4mm]
 \begin{pmatrix}
1 & -R(p^*)F^2(p) \E^{-2 t g(p)} \\
0 & 1
\end{pmatrix}, & p \in \mathcal{C}^*(\xi).\\[4mm]
\end{cases}
\]
Here $\id$ is the identity matrix. Moreover, the solutions of the initial and present RH problems for large $p$ 
are connected by 
\be \label{soed}m^{(4)}(p)=m(p)\begin{pmatrix} s^{-1}(p)&0\\0& s(p)\end{pmatrix}, \quad 
\mbox{where $s(p)=\E^{ t (\Phi(p) - g(p))}F(p)$}.
\ee \end{theorem}
Observe that for $p\in\mathcal C(\xi)\cup\mathcal C^*(\xi)$,
\[
0<\hat C \E^{-t \hat h(\xi)}\leq |F^{-2}(p)R(p)\E^{2tg(p)}|\leq C \E^{-t h(\xi)}, \quad 
\mbox{where $0<\hat h(\xi)<h(\xi)$},
\]
that is, as $t\to\infty$ the jump matrix $v^{(4)}(p)$ is exponentially close to the identity matrix on $\mathcal C(\xi)\cup\mathcal C^*(\xi)$.
The same holds true for $v^{(4)}(p)$ on $\mathcal C_1(\xi)\cup\mathcal C_1^*(\xi)$, 
except for a small neighborhood of the point $\gamma$. Since  
\[
|\chi(p)F_+^{-1}(p)F_-^{-1}(p)| +|\chi(p^*)F_-(p^*)F_+(p^*)|\leq C\E^{-t h(\xi, \varepsilon)},
\quad  p\in \mathcal C_1(\xi)\setminus \mathcal O_\varepsilon,
\]
where $h(\xi,\varepsilon)>0$ and  $\mathcal O_\varepsilon$ is a small neighborhood of $\gamma$, 
we conclude that outside of $\mathcal O_\varepsilon$,  $v^{(4)}(p)$ is exponentially 
close to the following matrix:
\be\label{modv}
v^{mod}(p)=
\left\{
 \begin{array}{ll}
\quad \I \id, & p \in \Sigma_{1,u}(\xi),\\
 \ -\I \id, & p \in \Sigma_{1,\ell}(\xi),\\[1mm]
\begin{pmatrix}
\E^{2\I tB+\I\Delta} & 0 \\
0 & \E^{-2\I tB-\I\Delta}
\end{pmatrix},& p \in I(\xi)\cup I^*(\xi).
\end{array}
\right.
\ee
Thus
we may expect that the solution of the RH problem for $m^{(4)}(p)$ can be approximated by the solution of the following model RH problem:
To find a holomorphic function in $\M(\xi)\setminus\big(\Sigma_1(\xi)\cup I(\xi)\cup I^*(\xi)\big)$ 
\be\label{moo}m^{mod}(p)=(m_1^{mod}(p), m_2^{mod}(p))
\ee
which is continuous up to the boundary except of possibly the points $-c-2d$ and $\gamma$, 
satisfies the jump condition
\be\label{moddd}
m_+^{mod}(p)=m_-^{mod}(p)v^{mod}(p)
\ee
with jump matrix \eqref{modv}, the standard symmetry condition
 \be\label{symmod}
 m^{mod}(p^*)=m^{mod}(p) \begin{pmatrix}0 & 1 \\ 1 & 0 \end{pmatrix},
 \ee 
and the normalization condition
\be\label{condmod}
 m_1^{mod}(\infty_+) m_2^{mod}(\infty_+)=1, \quad  m_1^{mod}(\infty_+)>0.
\ee
We will solve this problem explicitly in the next section.

\section{Solution of the model problem}
\label{sec:model}

We first choose a canonical basis of $\mathfrak a$ and $\mathfrak b$ cycles for the Riemann surface $\M(\xi)$ associated with the function \eqref{Rxi}.
The $\mathfrak b$ cycle surrounds the interval $[-c-2d, \gamma(\xi)]$ counterclockwise on the upper sheet $\Pi_U(\xi)$ and the $\mathfrak a$ cycle coincides with $I(\xi)\cup I^*(\xi)$, passing from $\gamma(\xi)$ to $-1$ on the upper sheet and back from $-1$ to $\gamma(\xi)$ on the lower sheet. Let $\omega_{p p^*}$, $p\in \M(\xi)\setminus (I(\xi)\cup I^*(\xi))$, be the normalized Abel differential
 of the third kind with poles at $p$ and $p^*$ such that
 \be\label{normomega}
 \int_{\mathfrak a} \omega_{p p^*}=0.
 \ee
 Then it can be represented as (see \cite{Teschl1})
\be\label{omm}
\omega_{pp^*} = \left(\frac{R^{1/2}(p,\gamma)}{\la - \pi(p)} + G(p) \right) \frac{d\la}{R^{1/2}(\la,\gamma)},
\ee
where $G(p)$ is determined from the normalization condition \eqref{normomega}.
Hence
\be\label{ggd}
G(p) = \frac{R^{1/2}(p,\gamma)}{Y}\int_{\mathfrak a}\frac{d\la}{(\pi(p)-\la)R^{1/2}(\la,\gamma)},
\ee
where
\be\label{normA}
Y=Y(\xi):=\int_{\mathfrak a}\frac{d\la}{R^{1/2}(\la,\gamma)}>0.
\ee

\begin{lemma}\label{omega2}
Uniformly for $\la\in\Sigma_1(\xi)$ there exists
\[
\lim_{p \rightarrow \infty_\pm} \omega_{pp^*}= \omega_{\infty_\pm \infty_\mp}.
\]
\end{lemma}

\begin{proof}
For example, let $p\in \Pi_U$. For $p\to\infty_+$,
\be\label{RR}
R^{1/2}(p,\gamma)=- p^2 +\frac{ \gamma-c-2d}{2}p + f(\gamma,p), \quad \mbox{ where $f(\gamma,p)= O\left(1\right)$.}
\ee
Then \eqref{ggd} and \eqref{RR} imply
\begin{align}
Y G(p)&=\int_{\mathfrak a}\frac{R^{1/2}(p,\gamma)}{p\big(1-\frac{\la}{p}\big) R^{1/2}(\la,\gamma)}d\la=\frac{R^{1/2}(p,\gamma)}{p} \int_{\mathfrak a}\frac{\big(1 +\frac{\la}{p} +\frac{\la^2}{p^2} \big)d\la}{R^{1/2}(\la,\gamma)}+
O(p^{-2})\nn \\
& =Y\,\frac{R^{1/2}(p,\gamma)}{p}\left(1+\frac{b}{p} +\frac{a}{p^2}\right) +O(p^{-2}), \label{gpg}
\end{align}
where the coefficients $a=a(\xi)$ and $b=b(\xi)$ are defined as 
\begin{align}\label{cond41}
  a:&=a(\xi)=\int_{I(\xi)}\frac{\la^2\,d\la}{R^{1/2}(\la,\gamma)}\left(\int_{I(\xi)}\frac{d\la}{R^{1/2}(\la,\gamma)}\right)^{-1},\\
\label{cond40}b:&=b(\xi)=\int_{I(\xi)}\frac{\la\, d\la}{R^{1/2}(\la,\gamma)}\left(\int_{I(\xi)}\frac{d\la}{R^{1/2}(\la,\gamma)}\right)^{-1}.
\end{align}Substituting \eqref{gpg} and \eqref{RR} in \eqref{omm} the following holds for any fixed $\la$ as $p\to\infty_+$
\begin{align} \nonumber
\omega_{p p^*}&=\frac{R^{1/2}(p,\gamma) + G(p)(\la - p) }{(\la - p)R^{1/2}(\la,\gamma)}=
-\frac{\frac{R^{1/2}(p,\gamma)}{p} +\big(\frac{\la}{p} - 1\big)G(p)}{\big(1 -\frac{\la}{p}\big)R^{1/2}(\la,\gamma)} \\ 
 \label{decomp3}
&=\frac{\la -b(\xi)}{R^{1/2}(\la,\gamma)} +\frac{1}{p}\frac{q(\la)}{R^{1/2}(\la,\gamma)} +
\frac{1}{p^2}\frac{f(\la,p)}{R^{1/2}(\la,\gamma)},
\end{align}
where
\be\label{PP}
q(\la):=q(\la,\xi)=\la^2 -\nu(\xi)\la +\nu(\xi) b(\xi) - a(\xi),\quad \nu(\xi)=\frac{\gamma(\xi) - c- 2d}{2}.
\ee
The function $f(\la,p,\xi)$ is uniformly bounded with respect to $(\la,p)$ on any compact
set for $\la$ as $p\to \infty_+$. 
Next, from condition \eqref{cond35} (b), we get
 \be\label{cond36}
 h:=h(\xi)=b(\xi).
 \ee We see
that the first summand in \eqref{decomp3} corresponds to the definition of $\omega_{\infty_+, \infty_-}$ (see \cite{tjac}). The case $p\to\infty_-$ is analogous.
\end{proof}

Equation \eqref{ggd} also implies the continuity of $G(p)$ on $\M(\xi)\setminus (I(\xi)\cup I^*(\xi))$, 
but along $I(\xi)\cup I^*(\xi)$ (oriented as above) the function $G(p)$ has a jump
\be\label{ggjump}
G_+(p)-G_-(p)=-\frac{2\pi \I}{Y},\quad p\in I(\xi)\cup I^*(\xi),
\ee
where the constant $Y$ is defined by \eqref{normA}.
To prove \eqref{ggjump} one applies the Sokhotski--Plemelj formula to \eqref{ggd}.
Similarly for any $\la\in\Sigma_1(\xi)\cup\Sigma$,
\[
\left(\omega_{p p^*}\right)_+ -\left(\omega_{p p^*}\right)_-
= -\frac{2\pi \I}{Y}\frac{d\la}{R^{1/2}(\la,\gamma)}=-2\pi\I\zeta,
\]
where
\[
\zeta= \frac{1}{Y}\frac{d\la}{R^{1/2}(\la,\gamma)}
\]
is the holomorphic Abel differential, normalized by the condition $\int_{\mathfrak a}\zeta=1$.
In summary, we proved the following

\begin{lemma}\label{omega3}
The function
\be\label{formF}
F(p):=\exp\left(\frac{1}{2\pi\I}\int_{\Sigma_1(\xi)}\log |\chi| \omega_{p p^*}\right)
\ee
solves the conjugation problem \eqref{ff} with
\be\label{Deltaj}
\Delta(\xi):=\I\int_{\Sigma_1(\xi)} \log |\chi|\zeta=
\frac{\I}{Y}\int_{\Sigma_1(\xi)} \frac{\log |\chi(\la)| d\la}{R^{1/2}(\la,\gamma)}\in\R.\ee
\end{lemma}

Note that this function has finite real limits at $\infty_\pm$ by Lemma \ref{omega2}. It satisfies
$F(p^*)=F^{-1}(p)$ because $\omega_{p p^*}=-\omega_{p^* p}$.
Moreover, \eqref{decomp3} implies that for $p\to \infty_+$
\be\label{decomp4}
F(p)=\exp\left(\frac{1}{2\pi\I}\int_{\Sigma_1(\xi)}\log |\chi| \omega_{\infty_+ \infty_-}
\right)\left( 1+ \frac{Q(\xi)}{p} +O\left(\frac{1}{p^2}\right)\right),
\ee
where
\be\label{QQ}
Q(\xi)=\int_{\Sigma_1(\xi)}\frac{\log |\chi(\la)|\,q(\la,\xi)d\la}{2\pi\I\,R^{1/2}(\la,\gamma)}\in\R,
\ee 
with $q(\la,\xi)$  defined by  \eqref{PP}.

Denote by
\[
\tau=\tau(\xi)=\int_{\mathfrak b}\zeta=- \frac{1}{Y}\int_{\Sigma_1(\xi)}\frac{d\la}{R^{1/2}(\la,\gamma)}
\]
the $\mathfrak b$-period of the normalized holomorphic Abel differential $\zeta$. Since $Y>0$
by \eqref{normA} we have $\tau\in\I \R_+$. Introduce the theta function
\[
\theta(v):=\theta(v\,|\,\tau)=\sum_{m\in\Z} \exp\big(\pi\I m^2\tau + 2\pi\I m v\big)
\]
and the Abel map $A(p):=A(p,\xi)=\int_{-c-2d}^p \zeta$, which has the following properties (cf. \cite{FK,tjac}):
\begin{align} \nonumber
& A(p^*)= - A(p), \qquad  A(\gamma)=-\frac{\tau}{2} \ (\mbox{mod $\tau$}), \\ \nonumber
& A_+(p)-A_-(p)=-\tau \quad \mbox{for} \quad p\in I(\xi)\cup I^*(\xi),\\ \label{new8}
& A(p)=A(\infty_\pm) \pm\frac{1}{Y \la} + O(\la^{-2}), \quad p=(\la,\pm)\to\infty_\pm, \\ \label{proppp}
& 2A(\infty_+)=A(\infty_+) - A(\infty_-)=\frac{1}{2\pi\I}\Lambda,
\end{align}
where
\be\label{omper} \Lambda=\Lambda(\xi)=\int_{\mathfrak b}\omega_{\infty_+ \infty_-}.
\ee
Let $\Xi=\frac{\tau}{2} + \frac{1}{2}$ be the Riemann constant. Then the functions
$\theta\big(A(p) +\frac{\tau}{2} - \Xi\big)$ and $\theta\big(A(p^*) +\frac{\tau}{2} - \Xi\big)$ both have their
only zeros at $\gamma$. Denote
\be\label{alphaa}
\alpha(p):=\alpha(p,\xi)=
\frac{\theta\big(A(p) +\frac{\tau(\xi)}{2} - \Xi(\xi)+\frac{t B(\xi)}{\pi} +\frac{\Delta(\xi)}{2\pi}\big)}{\theta\big(A(p) +\frac{\tau(\xi)}{2} - \Xi(\xi)\big)}.
\ee
Standard properties of theta functions show that this
function is holomorphic on $\M(\xi)$ except at the point $\gamma$, and has a single jump along $I(\xi)\cup I^*(\xi)$,
\begin{align*}
\alpha_+(p)& =\alpha_-(p) \E^{2 \I t B(\xi) + \I\Delta(\xi)}, \\
\alpha_+(p^*)& =\alpha_-(p^*) \E^{-2 \I t B(\xi) - \I\Delta(\xi)}.
\end{align*}
Now, on the upper sheet $\Pi_U(\xi)$, introduce the function
\be\label{delta}
\delta(p)=\sqrt[4]{\frac{\pi(p) -\gamma}{\pi(p) + c +2d}},
\ee
where the branch of $\sqrt[4]{\cdot}$ is chosen to take positive values for $\pi(p)>\gamma$. Hence this function also takes  positive values for $p<-c-2d$.
We continue $\delta(p)$ by $\delta(p^*)=\delta(p)$ as an even function to $\Pi_L(\xi)$. Observe that $\delta(p)$ has no jump on $\Sigma$,
because it is real-valued and takes equal values at symmetric points of $\Sigma_u$ and $\Sigma_\ell$ on $\Pi_U$. On $\Sigma_1(\xi)$, it solves the following conjugation problem:
\[
\delta_+(p)=\delta_-(p)
\begin{cases}
\I, & p \in \Sigma_{1,u}(\xi),\\[1mm]
 -\I,& p\in \Sigma_{1,\ell}(\xi).
\end{cases}
\]
Thus the vector $\hat m(p)=\left(\delta(p)\alpha(p), \delta(p^*)\alpha(p^*)\right)$ solves the jump problem
$\hat m_+(p)=\hat m_-(p) v^{mod}(p)$, where $v^{mod}(p)$ is defined by \eqref{modv}, and satisfies the symmetry condition \eqref{symto}.
But $\hat m(p)$ does not satisfy the normalization condition, because $\hat m_1 (\infty_+)
\hat m_2 (\infty_+)\neq 1$. To mend this, recall that $\delta(\infty_+)=\delta(\infty_-)=1$. Hence
\be\label{mmodf}
m^{mod}(p)=\frac{1}{\sqrt{\alpha(\infty_+)\alpha(\infty_-)}}
\left(\delta(p)\alpha(p), \delta(p^*)\alpha(p^*)\right)
\ee
solves the model problem including \eqref{condmod}. Note that this solution is bounded everywhere in $\M(\xi)$ except at the branch points $-c-2d$, $\gamma(\xi)$, where both components of $m^{mod}(p)$ have singularities of type $O\big((p+c+2d)^{-1/4}\big)$ and $O\big((p-\gamma)^{-1/4}\big)$. The first singularity matches well with \eqref{resonance}. 

Our next task is to derive the asymptotic formula for this vector as $p\to\infty_+$ up to terms of order $O(p^{-2})$.

\begin{lemma}\label{mainn}
Let $\Omega_0$ be the Abel differential of the second kind on $\M(\xi)$ with second order poles at $\infty_+$ and $\infty_-$ normalized by the condition $\int_{\mathfrak a}\Omega_0=0$ and let $\omega_{\infty_+\infty_-}$ be
the Abel differential of the third kind as above. Then
\[
\frac{\pa}{\pa p} g(p, \xi)d\pi=\Omega_0 + \xi\omega_{\infty_+\infty_-},
\]
where $\frac{\pa}{\pa p} g(p, \xi)$ is defined by \eqref{defgee}.
Moreover,
\[
- 2\I B(\xi)=U +\xi\Lambda,
\]
where $U= U(\xi)=\int_{\mathfrak b}\Omega_0$ is the $\mathfrak b$-period of $\Omega_0$ and $\Lambda=\Lambda(\xi)$ is the
$\mathfrak b$-period of $\omega_{\infty_+\infty_-}$ defined by \eqref{omper}.
\end{lemma}

This lemma is a direct corollary of Lemma~\ref{abelint} and \eqref{cond37}.
Since $\xi=\frac{n}{t}$ we have
\[
\frac{t B(\xi)}{\pi}=-t\frac{U}{2\pi\I} - n\frac{\Lambda}{2\pi\I}
\]
and substituting this into \eqref{alphaa}, taking into account that $\theta(v+1)=\theta(v)$, we get
\be\label{new7}
\alpha(p)=
\frac{\theta\big(A(p) +\frac{\tau(\xi)}{2} - \Xi(\xi)- t\frac{U}{2\pi\I} - n\frac{\Lambda}{2\pi\I} +\frac{\Delta(\xi)}{2\pi}\big)}{\theta\big(A(p) +\frac{\tau(\xi)}{2} - \Xi(\xi)\big)}.
\ee
Passing to the limit as $p\to\infty_+$ and using property \eqref{proppp} we obtain
\[
\alpha(\infty_+)=\frac{\theta(\underline{z}(n-1,t))}{\alpha^+},\quad \alpha(\infty_-)=\frac{\theta(\underline{z}(n,t))}{\alpha^-},
\]
where
\be\label{defzz} \underline{z}(n,t):=A(\infty_+)- n\frac{\Lambda(\xi)}{2\pi\I}-t\frac{U(\xi)}{2\pi\I} +\frac{\tau(\xi)}{2}
+\frac{\Delta(\xi)}{2\pi}-\frac{\Lambda(\xi)}{2\pi\I}- \Xi(\xi),
\ee
and
\[
\alpha^\pm:=\alpha^\pm(\xi)=\theta\Big(A(\infty_\pm) +\frac{\tau(\xi)}{2} - \Xi(\xi)\Big).
\]
Note that since $\frac{\Delta(\xi)}{2\pi}-\frac{\Lambda(\xi)}{2\pi\I}\in\R$ and $\frac{\tau(\xi)}{2}=-A(\gamma)$,
by the Jacobi inversion theorem (\cite{tjac}), there exists a point $\rho(\xi)\in I(\xi)\cup I^*(\xi)$ such that
\be\label{divizor}
-A(\rho)+A(\gamma)=\frac{\Delta(\xi)}{2\pi}-\frac{\Lambda(\xi)}{2\pi\I}.
\ee
Thus we can represent $\underline{z}(n,t)$ in a more familiar form for finite-gap operators
 \be\label{import6}
 \underline{z}(n,t):=A(\infty_+)-A(\rho)- n\frac{\Lambda}{2\pi\I}-t\frac{U}{2\pi\I} - \Xi.
 \ee
Passing to the limit $p\to\infty_+$ in the first component of vector \eqref{mmodf} we get
\be\label{sss}
m_1^{mod}(\infty_+)=\beta(\xi)\,\sqrt{\frac{\theta(\underline{z}(n-1,t))}{\theta(\underline{z}(n,t))}},
\ee
where
\be\label{import1}
\beta(\xi)=\sqrt{\frac{\alpha^-}{\alpha^+}}=\sqrt{\frac{\theta\big(A(\infty_-) +\frac{\tau(\xi)}{2} - \Xi(\xi)\big)}{\theta\big(A(\infty_+) +\frac{\tau(\xi)}{2} - \Xi(\xi)\big)}}\in \R
\ee
is a real-valued continuously differentiable function of $\xi$ with a bounded derivative on the interval under consideration.
Recall now that we assume $m^{(4)}(p)=m^{mod}(p) +o(1)$ as $p\to\infty_+$, where $o(1)$ is understood with respect to $t\to\infty$. Apply this to
 \eqref{sss}, \eqref{soed}, \eqref{Kxi}, \eqref{formF}, \eqref{eqas5}, and Lemma \ref{omega2}.
We obtain that the first component of the solution of the initial RH problem, which is equal to
$\prod_{j=n}^{+\infty}(2 a(j,t))$ by \eqref{AB} and \eqref{asm}, is
\[
\prod_{j=n}^{+\infty}(2 a(j,t))=\beta(\xi)\E^{t K(\xi)}\E^{\frac{1}{2\pi\I}\int_{\Sigma_1(\xi)}\log |\chi|\omega_{\infty_+\infty_-}}\,
\sqrt{\frac{\theta(\underline{z}(n-1,t))}{\theta(\underline{z}(n,t))}}\big(1+ o(1)\big),
\]
where $\beta(\xi)$ and $\underline{z}(n,t)$ are defined by \eqref{import1} and \eqref{import6}, and
$o(1)$ is a function which tends to $0$ as $t\to \infty$.
Recall that we treat $\xi$ as a slow variable, and that all functions of $\xi$ are continuously differentiable with respect to $\xi$ in our considerations.
This means, for example, that if $\xi=\frac{n}{t}$ and $\xi^\prime=
\frac{n+1}{t}$, then $\beta(\xi)\beta^{-1}(\xi^\prime)=1 +O(t^{-1})$ and
\[
\E^{t(K(\xi)-K(\xi^\prime))}=\E^{-\frac{d K}{d \xi}(\xi)}+ O(t^{-1}),
\]
etc.
Consequently,
\be\label{asan}
a^2(n,t)=\frac{1}{4}\E^{-2 K^\prime(\xi)}\,\frac{\theta(\underline{z}(n-1,t))\theta(\underline{z}(n+1,t))}{\theta^2(\underline{z}(n,t))} +o(1).
\ee
\begin{lemma}\label{et}
Let $\tilde a=\tilde a(\xi)$ be the constant from the expansion (cf.\ \cite[Eq.\ (9.42)]{tjac})
\[
\E^{\int_{-c-2d}^p \omega_{\infty_+ \infty_-}}=-\frac{\tilde a}{\la}\bigg(1 +\frac{\tilde b}{\la} 
+O(\la^{-2})
\bigg),\quad p=(\la,+)\in\Pi_U(\xi),
\]
and let $K^\prime(\xi):=\frac{d K}{d \xi}(\xi)$, $k^\prime(\xi)=\frac{d k}{d \xi}(\xi)$, with 
$K(\xi)$ and $k(\xi)$ defined by \eqref{Kxi} and \eqref{eqas5}, respectively. Then
\begin{align}\label{zhut}
  \frac{1}{4}\E^{-2 K^\prime(\xi)}&=\tilde a^2(\xi),\\
\label{zhut3} k^\prime(\xi)&=\tilde b(\xi).
\end{align}
\end{lemma}

\begin{proof}
 According to equation (9.43) in \cite{tjac}, the constant $\tilde a$
 can be computed as
 \[
 \log \tilde a=\lim_{\la\to +\infty}\bigg(\int_1^\la \omega_{\infty_+ \infty_-} +\log \la\bigg).
 \]
 Therefore, one has to prove that $K^\prime(\xi) +\log 2=-\log \tilde a$. By \eqref{Kxi} and \eqref{deri},
 \begin{align}
 & \frac{d}{d\xi} \nonumber
 \lim_{\la\to +\infty}\int_1^\la\Bigg(\frac{(x-\mu(\xi))\sqrt{x-\gamma(\xi)}} {\sqrt{(x^2 - 1)(x+c+2d)}}-\frac{x+\xi}{\sqrt{x^2-1}}\Bigg)dx  +\log 2\\
 \nonumber
 &\quad =-\lim_{\la\to +\infty}\bigg(\int_1^\la \frac{\pa^2}{\pa\xi\pa x} \tilde g(x,\xi) dx +\log\left(\la +\sqrt{\la^2 - 1}\right)\bigg) +\log 2\\ \nonumber
& \quad = -\lim_{\la\to +\infty}\bigg(\int_1^\la \omega_{\infty_+ \infty_-} +\log \la\bigg),
\end{align}
which proves \eqref{zhut}.
 Formula \eqref{zhut3} follows immediately from \eqref{deri}. \end{proof}
 To derive the term of order $\la^{-1}$ for the first component of the solution we observe
 that by \eqref{new8}, \eqref{new7}, and \eqref{defzz},
\be\label{decomp12}
\alpha(p) = \alpha(\infty_+) \left(1+\frac{1}{Y\la}\,\frac{\pa}{\pa w}\log\left(
\frac{\theta\big(\underline z(n-1,t) + w\big)}{\theta\big(A(\infty_+) +\frac{\tau}{2} - \Xi+w\big)}\right)\Big|_{w=0} +O(\la^{-2})\right),
\ee
$p=(\la,+)$. Denote
\be\label{eta1}
\eta:=\eta(\xi)=-\frac{1}{Y}\frac{\pa}{\pa w}\log\left(\theta\Big(A(\infty_+) +\frac{\tau}{2} - \Xi+w\Big)\right)\Big|_{w=0} -\frac{\gamma + c +2d}{4}.
\ee
Then combining \eqref{mmodf}, \eqref{sss}, \eqref{delta},
\eqref{decomp12}, and \eqref{eta1} we get
\[
m_1^{mod}(p)=m_1^{mod}(\infty_+)\left(1+
\frac{1}{Y\la}\,\frac{\pa}{\pa w}\log\theta\big(\underline z(n-1,t) + w\big)\Big|_{w=0}
+ \frac{\eta}{\la}+O(\la^{-2})\right).
\]
Respectively, by \eqref{defdg}, \eqref{m1}, \eqref{decomp4}, \eqref{QQ}, \eqref{PP}, \eqref{eqas5}, and \eqref{asm}
\[
-B(n-1,t)=\Big(-t k(\xi) +\frac{t}{2}+\eta(\xi) + Q(\xi) + \frac{1}{Y}\,\frac{\pa}{\pa w}\log
\theta\big(\underline z(n-1,t) + w\big)|_{w=0}\Big)(1 +o(1)),
\]
where $o(1)$ tends to $0$ as $t\to\infty$ and $\frac{n}{t}=\xi$ is almost a constant.
Apply now the same arguments as for \eqref{asan} and use \eqref{zhut3} to get
$b(n,t)=b(n,t,\xi) +o(1)$, where
\be\label{bee}
b(n,t,\xi)=\tilde b +\frac{1}{Y}\,\frac{\pa}{\pa w}\log\left(\frac{
\theta\big(\underline z(n-1,t) + w\big)}{\theta\big(\underline z(n,t) + w\big)}\right)\Big|_{w=0}.
\ee
Lemma \ref{et} and \eqref{asan} also imply $a^2(n,t)=a^2(n,t,\xi) +o(1)$ with
\be\label{aee}
a^2(n,t,\xi)=\tilde a^2\,\frac{\theta(\underline{z}(n-1,t))\theta(\underline{z}(n+1,t))}{\theta^2(\underline{z}(n,t))}.
\ee
Formulas \eqref{aee} and \eqref{bee} for fixed $\xi$ describe the standard two-band Toda lattice solution
(cf.\ \cite[Theorem 9.48]{tjac}). 
Thus, assuming that the solution of the model RH problem \eqref{modv}--\eqref{condmod} approximates the solution
of the initial RH problem I--IV well, we arrive at the following

\begin{theorem} \label{theor1}
Let
\begin{enumerate}[1.]
 \item
$\{a(n,t),\,b(n,t)\}$ be the solution of the problem \eqref{tl}, \eqref{ini}--\eqref{cont12}
 as $n, t\to\infty$ in the domain $\xi_{cr}^\prime\, t<n<\xi_{cr}\, t$, where the parameters $\xi_{cr}^\prime$ and $\xi_{cr}$ are defined by \eqref{xicr}, \eqref{mui}, and \eqref{xicrit2};
 \item $\xi\in(\xi_{cr}^\prime,\,\xi_{cr})$ be a  parameter; $\gamma(\xi)\in (-c-2d, -c+2d)$ be defined by \eqref{eqas} and \eqref{zerocond}; $\M(\xi)$ be the Riemann surface of the function \eqref{Rxi};
 \item  $\rho(\xi)$ be a  point on  $\M(\xi)$, defined by  \eqref{divizor}, \eqref{Deltaj} with $\pi(\rho)\in [\gamma, -1]$;
 \item $\{a(n,t,\xi), b(n,t,\xi)\}$ be the finite-gap solution corresponding to the initial divisor $\rho(\xi)$ via \eqref{aee}, \eqref{bee}, \eqref{import6}.
\end{enumerate}
Then in a vicinity of any ray $n=\xi t$ the solution of the problem \eqref{tl}, \eqref{ini}--\eqref{cont12} has the following asymptotical behavior as $t\to +\infty$
 \be\label{mainab}
 a(n,t)=a(n,t,\xi) +o(1),\quad b(n,t)=b(n,t,\xi) +o(1).
 \ee
\end{theorem}

\section{Asymptotics of the solution in the domain \texorpdfstring{$\xi_{cr,1} < \xi < \xi_{cr,1}^\prime$}{xi,cr,1< xi <xi,prime,cr,1}}
\label{sec:el}

In this domain we study the asymptotic behavior of the solution to the problem \eqref{tl}--\eqref{cont12} with the help of the RH problem I--IV. The considerations are similar to those in Sec.\ \ref{secg}--\ref{sec:model}, and we give a short description of the necessary changes. Let $\xi \in (\xi_{cr,1}, \xi_{cr,1}^\prime)$ with
$\xi_{cr,1}$ and $\xi_{cr,1}^\prime$ defined by
\eqref{xicrit1}, \eqref{mui}, \eqref{xicr}. Here we choose the $g$-function with its moving point $\gamma_1=\gamma_1(\xi)$ on the interval $(-1,1)$,
\[
g_1(p,\xi)= - \int_{1}^p \frac{(x - \mu_1(\xi))(x-\gamma_1(\xi))}{R^{1/2}(x,\gamma_1(\xi))} dx,
\]
where the function
\be \label{Rxi1}
R_1^{1/2}(\la,\gamma_1)=-\sqrt{(\la+c+2d)(\la + c - 2d)(\la-1)(\la-\gamma_1)}
\ee
defines the corresponding Riemann surface $M_1(\xi)$.
The points $\mu_1(\xi) \in (-c+2d, \gamma_1(\xi))$ and $\gamma_1(\xi)\in (-1,1)$ are chosen to satisfy 
\be \label{zerocond1}
\int_{-c+2d}^{\gamma_1(\xi)}\frac{(\la - \mu_1(\xi))(\la - \gamma_1(\xi))}{R_1^{1/2}(\la, \gamma_1(\xi))}d\la = 0, \quad
2c - 1 + \gamma_1(\xi) + 2 \mu_1(\xi)=-2\xi.
\ee
The last condition implies
\[
\Phi_1(p, \xi)- g_1(p,\xi)=K_1(\xi) + \frac{k_1(\xi)}{\la} +  O(\la^{-2}),
\quad \mbox{ as $p=(\la, +)\to \infty_+$},
\]
where $K_1$ and $k_1$ are real-valued constants.
For the same reasons as above,
\be \label{partial K_1}
\frac{\partial K_1(\xi)}{\partial \xi}= \lim_{p \to \infty_+}\Big( -\log(z_1(p)) +
\int_1^{p} \omega_{\infty_+ \infty_-}^1 \Big) = \log\tilde a_1 - \log d,
\ee
where $z_1(p)$ is defined by \eqref{z_1} and $\omega_{\infty_+ \infty_-,1} $ is the normalized Abel integral of the third kind on $M_1(\xi)$. The $\mathfrak b_1$ circle surrounds the interval $[-c-2d, -c+2d]$ counterclockwise on the upper sheet, and $\mathfrak a_1$ follows the gap $I^1(\xi)=[-c+2d, \gamma_1]$ on the upper sheet in positive direction, and then back on the lower sheet of $\M_1(\xi)$.

The following lemma justifies the existence of this $g$-function; its proof is the same as before.
\begin{lemma} \label{gest1}
The functions $\gamma_1(\xi) \in (-1,1)$ and $\mu_1(\xi) \in (-c+2d, \gamma_1(\xi))$ satisfying \eqref{zerocond1} exist for $\xi\in (\xi_{cr,1}, \xi_{cr,1}^\prime)$. On this interval,
$\gamma_1(\xi)$ is decreasing with $\gamma_1(\xi_{cr,1})=1$ and $\gamma_1(\xi_{cr,1}^\prime)=-1$.
\end{lemma}

The signature table for $\re g_1$ is depicted in Figure~\ref{fig:Reg1}.
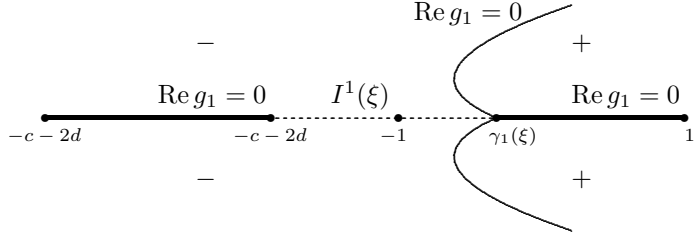
\begin{figure}[ht]
\begin{picture}(9,3.3)

\linethickness{0.6mm}
\put(0.5,1.5){\line(1, 0){3}}
\put(6.5,1.5){\line(1, 0){2.5}}
\put(0.5,1.5){\circle*{0.1}}
\put(3.5,1.5){\circle*{0.1}}
\put(0,1.2){\scriptsize $-c-2d$}
\put(3,1.2){\scriptsize $-c-2d$}
\put(2.5,2.4){$-$}
\put(2.5,0.6){$-$}
\put(7.5,0.6){$+$}
\put(7.5,2.4){$+$}
\put(2,1.7){$\re g_1 =0$}
\put(7.5,1.7){$\re g_1 =0$}
\put(5.4,2.8){$\re g_1 =0$}
\put(4.3,1.7){$I^1(\xi)$}

\linethickness{0.2mm}
\curve(6.5,1.5, 6,2.2, 7.5,3)
\curve(6.5,1.5, 6,0.8, 7.5,0)
\put(6.5,1.5){\circle*{0.1}}
\put(9,1.5){\circle*{0.1}}
\put(6.4,1.2){\scriptsize $\gamma_1(\xi)$}
\put(9,1.2){\scriptsize $1$}

\curvedashes{0.05,0.05}
\curve(3.5,1.5, 6.5,1.5)
\put(5.2,1.5){\circle*{0.1}}
\put(4.95,1.2){\scriptsize $-1$}

\end{picture}
  \caption{Sign of $\re g_1(p)$ on the upper sheet of $\M_1(\xi)$}\label{fig:Reg1}
\end{figure}
The remaining considerations are completely identical to those of the previous two sections up to a natural symmetry,
so we just formulate the result here. Denote
\[
\int_{-c-2d}^{-c+2d}\frac{(\la - \mu_1)(\la-\gamma_1)}{R_1^{1/2}(\la+\I 0, \gamma_1)} d\la = \I B_1(\xi),
\]
and 
\be\label{Deltaj1}
\Delta_1(\xi):=
\frac{\I}{Y_1}\int_{\Sigma(\xi)} \frac{\log |\chi(\la)| d\la}{R_1^{1/2}(\la,\gamma_1)},\quad Y_1=2\int_{I^1}\frac{d\la}{R_1^{1/2}(\la,\gamma_1)}.
\ee
We preserve the clockwise orientation on $\Sigma_1$ and on the truncated contour $\Sigma(\xi)$.
\noprint{After Step 4, we arrive at the following model problem: Find a holomorphic vector-function
\eqref{moo}
in the domain $\M_1(\xi)\setminus\big(\Sigma(\xi)\cup I^1(\xi)\cup I^{1,*}(\xi)\big)$, bounded at $\infty_+$, $\infty_-$,
and satisfying the jump condition
\eqref{moddd}
with the jump matrix
\be \label{modv1}
v^{mod}(p) = \begin{cases}
 -\I \id, & p \in \Sigma_{u}(\xi),\\[1mm]
\I \id, & p \in \Sigma_{\ell}(\xi),\\[1mm]
\begin{pmatrix}
\E^{-2 \I t B_1(\xi) + \I \Delta_1(\xi)} & 0 \\
0 & \E^{2\I t B_1(\xi) - \I \Delta_1(\xi)}
\end{pmatrix},& p \in I^1(\xi)\cup I^{1,*}(\xi),\\[4mm]
\end{cases}
\ee
and the symmetry  and the normalization conditions .
The solution of this model problem can be derived similarly to the one in Section \ref{sec:model}. 
and we describe it briefly.}

Let $\tilde\omega_{p p^*}$ be the normalized Abel differential
 of the third kind with poles at $p$ and $p^*$ and $\tilde\zeta$ be the holomorphic normalized Abel differential on $\M(\xi)$.
Denote by $\tau_1=\int_{\mathfrak b_1}\tilde \zeta$ and let $\theta_1(v)=\theta(v\,|\,\tau_1)$.
Set $A_1(p)=\int_{-c-2d}^p \tilde \zeta$ and denote by $\Xi_1$ the Riemann constant.
\noprint{Then
$\theta_1\big(A_1(p) +\frac{\tau_1}{2} - \frac{1}{2} - \Xi_1\big)$ has a zero at  $\gamma_1$.
Denote
\be\label{alphaa1}
\alpha_1(p):=\alpha_1(p,\xi)=
\frac{\theta_1\big(A_1(p^*) +\frac{\tau_1(\xi)}{2} - \frac{1}{2} - \Xi_1(\xi) + \frac{t B_1(\xi)}{\pi} -\frac{\Delta_1(\xi)}{2\pi}\big)}{\theta_1\big(A_1(p^*) +\frac{\tau_1(\xi)}{2} -\frac{1}{2} - \Xi_1(\xi)\big)}
\ee
and
$\delta_1(p)=\sqrt[4]{\frac{\pi(p) -\gamma_1}{\pi(p) -1}}.
$
The vector
\[
m^{mod}(p)=\frac{1}{\sqrt{\alpha_1(\infty_+)\alpha_1(\infty_-)}}
\left(\delta_1(p)\alpha_1(p), \delta_1(p^*)\alpha_1(p^*)\right)
\]
solves the model problem \eqref{moddd}--\eqref{condmod} with jump matrix \eqref{modv1}.}
Let $U_1$ be the $\mathfrak b$-period of the Abel differential of the second kind on $\M_1(\xi)$ with
second order poles at $\infty_+$ and $\infty_-$ and $\Lambda_1$ be the $\mathfrak b$-period of $\tilde \omega_{\infty_+\infty_-}$. We use
\[
\frac{t B_1(\xi)}{\pi}=-t\frac{U_1}{2\pi\I} - n\frac{\Lambda_1}{2\pi\I}
\] and define
\be\label{defzz1}
\underline{z}_1(n,t):=A(\infty_+)- n\frac{\Lambda_1}{2\pi\I}-t\frac{U_1}{2\pi\I} +
\frac{\tau_1(\xi)}{2} -\frac{1}{2} -\frac{\Delta_1(\xi)}{2\pi}-\frac{\Lambda_1}{2\pi\I}-\Xi_1(\xi).
\ee
For $\xi$ fixed and the spectrum $[-c-2d, -c+2d]\cup [\gamma_1(\xi), 1]$, let $\{a_1(n,t,\xi), b_1(n,t,\xi)\}$  be the finite-gap solution of the Toda lattice defined by 
\begin{align} \label{aee1}
(a_1(n,t,\xi))^2&=
\tilde a_1^2 \frac{\theta_1(\underline{z}_1(n+1,t))\theta_1(\underline{z}_1(n-1,t))}
{\theta_1^2(\underline{z}_1(n,t))}, \\ \label{bee1}
b_q(n,t,\xi) & =\tilde b +\frac{1}{\Gamma_1}\frac{\pa}{\pa w}\log\left(\frac{
\theta\big(\underline z_1(n-1,t) + w\big)}{\theta\big(\underline z_1(n,t) + w\big)}\right)\Big|_{w=0}.
\end{align}
This finite-gap solution is connected with the solution of the model problem in the same way as in the previous section. Thus, if we solve the parametrix problem, we can expect that the following result is valid. \begin{theorem}\label{theor2} Let
\begin{enumerate}[1.]
 \item
$\{a(n,t),\,b(n,t)\}$ be the solution of the problem \eqref{tl}, \eqref{ini}--\eqref{cont12}
 as $n, t\to\infty$ in the domain $\xi_{cr,1}\, t<n<\xi_{cr,1}^\prime\, t$, where
 $\xi_{cr,1}$ and $\xi_{cr,1}^\prime$ are defined by \eqref{xicrit1}--\eqref{xicr};
 \item $\xi\in(\xi_{cr,1},\, \xi_{cr,1}^\prime)$ be a  parameter; $\gamma_1(\xi)\in (-1, 1)$ be
 defined by \eqref{zerocond1};  \item $\{a_1(n,t,\xi), b_1(n,t,\xi)\}$ be the finite-gap solution defined by 
 \eqref{defzz1}--\eqref{bee1}.
\end{enumerate}
Then in a vicinity of any ray $n=\xi t$ the solution of the problem \eqref{tl}, \eqref{ini}--\eqref{cont12} has the following asymptotical behavior as $t\to +\infty$
\[
 a(n,t)=a_1(n,t,\xi) +o(1),\quad b(n,t)=b_1(n,t,\xi)  +o(1).
\]
\end{theorem}

\begin{remark} In the case of spectra of equal length, that is, if $d=1/2$, one could also derive these results 
	from Section~\ref{sec:red} by the following transformation. Set $\hat a(n):=a(-n-1)$, $\hat b(n):=-b(-n)$ and rescale the spectral parameter by $\hat \la:= (\la - c)/2d$, which shifts the original problem with background spectra $[-c-2d, -c+2d]$ and $[-1,1]$ to $[-\hat c-2\hat d, -\hat c+2 \hat d]$ and $[-1,1]$, where $\hat c= c/2d$, $\hat d=1/4d$. To obtain explicit 
formulas for the solution in terms of theta functions, this approach requires to recompute all quantities, for example, $\hat \xi_{cr}=\frac{1}{2d}\xi_{cr,1}=\xi_{cr,1}$, and so on. 
\end{remark}

\section{Asymptotics of the solution for
\texorpdfstring{$\xi \in (\xi_{cr,1}^\prime, \xi_{cr,0}) \cup (\xi_{cr,0}, \xi_{cr}^\prime)$}{xi,prime,cr,1<n/t<xi,prime,cr,2}}
\label{sec:gap}

In this region we work on the initial Riemann surface $\M$ corresponding to the function \eqref{R1/2}. The
$g$-functions associated with the left and right RH problems \eqref{eq:jumpcond} and \eqref{eq:jumpcond2} coincide up to the sign. Namely, introduce two points $-c+2d<\mu_1(\xi) < \mu_2(\xi) <-1$ such that
\be \label{condgap}
\int_{-c+2d}^{-1}\frac{(\la - \mu_1(\xi))(\la - \mu_2(\xi))}{R^{1/2}(\la)}d\la = 0, \quad 
\mu_1(\xi) +  \mu_2(\xi)+c=-\xi.
\ee
The last equality implies that
\[
 g_{gap}(p, \xi)=  \int_{1}^p \frac{(\zeta - \mu_1(\xi))(\zeta-\mu_2(\xi))}{R^{1/2}(\zeta)} d\zeta
\]
has the following asymptotical behavior
\[
 g_{gap}(p)=\Phi(p) + O(1)=-\Phi_1(p) + O(1)  \quad \mbox{ as $p\to\infty_\pm$}.
\]
Moreover, $\mu_1$ and $\mu_2$ exist as $\xi\in (\xi_{cr,1}^\prime, \xi_{cr}^\prime)$ and
the line $\re g_{gap}(p)=0$ crosses the real axis between these two points.
We introduce the $\mathfrak b$ and $\mathfrak a$ cycles and all normalized Abel differentials on $\M$
analogously to those on $\M(\xi)$, $\M_1(\xi)$. It is evident that both RH problems can be solved by an analogous procedure as above, but with Step 1 (cf.\ Section~\ref{sec:red}) excluded. Moreover, to cancel the poles at $p_0$ 
and $p_0^*$ when performing Step 2 one has to use the right RH problem  \eqref{eq:jumpcond}, \eqref{polecond} for $\xi\in (\xi_{cr,0},\xi_{cr}^\prime)$, and the left RH problem \eqref{eq:jumpcond2}, \eqref{polecond1} for $\xi\in(\xi_{cr,1}^\prime, \xi_{cr,0})$. Thus, as the asymptotic of the solution of the problem \eqref{tl}--\eqref{cont12} we get two finite-gap solutions, described by \eqref{aee}--\eqref{bee} and \eqref{aee1}--\eqref{bee1} with the same theta function ($\tau=\tau_1$), but with different arguments: \eqref{defzz} and \eqref{defzz1}. Namely, since $\ul z(n,t)$ and $\ul z_1(n,t)$ are defined on the same surface, then $\Lambda=\Lambda_1$, $U=U_1$, $\Xi=\Xi_1$, $A(\infty_+)=A_1(\infty_+)$. Moreover, none of them depends on $\xi$. 
Thus
\be\label{zetasv}
\ul z_1(n,t)=\ul z(n,t)-\frac{\Delta}{2\pi} + \frac{1}{2} - \frac{\Delta_1}{2\pi},
\ee
where
\[
\Delta:=\I\int_{\Sigma_1} \log |\chi|\zeta,\quad \Delta_1:=\I\int_{\Sigma} \log |\chi|\zeta,
\]
and $\zeta$ is the holomorphic normalized Abel differential on $\M$. The contours  $\Sigma_1$ and $\Sigma$ are the same as in Section 2. Formula \eqref{zetasv} implies the following reasoning. To prove that the asymptotics of the Toda lattice solution are the same when obtained from the right and from the left RH problems in the absence of discrete spectrum in the gap $[-c+2d, -1]$, it is sufficient to prove that
$\Delta+\Delta_1= \pi\,(\mbox{mod}\, 2\pi)$ or
 \be\label{oo} \E^{\I(\Delta+\Delta_1)}=-1. \ee
Let us first make sure that this is true. Denote $I=I^1=[-c+2d, -1]\in\Pi_U$.
Recall that (cf.\ Lemma \ref{omega3})
the function
\[
F(p):=\exp\left(\frac{1}{2\pi\I}\int_{\Sigma\cup\Sigma_1}\log |\chi| \omega_{p p^*}\right)\]
is the unique solution of the following conjugation problem: to find a holomorphic function on  $\M\setminus(\Sigma\cup\Sigma_1\cup I\cup  I^*)$,  such that
\begin{align} \label{condu}
F(p^*)&=F^{-1}(p),\quad p\in\M, \\ \label{ffff}
F_+(p)&=F_-(p)  |\chi(p)|, \quad p \in \Sigma\cup\Sigma_1.
\end{align}
On the set $ I\cup I^*$ this function has a jump,
\be\label{DDD}
F_+(p)=F_-(p)\E^{\I \tilde\Delta}, \quad  p\in  I\cup  I^*, \quad \mbox{with }\,
\tilde\Delta:=\I\int_{\Sigma\cup\Sigma_1} \log |\chi|\zeta=\Delta +\Delta_1.
 \ee
The orientation on $ I\cup I^*$ is the same as for the $\mathfrak a$ cycle.
Note that the jump along $  I\cup  I^*$ can not be arbitrary.
In fact, one can solve \eqref{ffff} without the Sokhotski--Plemelj formula as follows.
Consider the function
\[
\tilde\delta(\la)=\sqrt[4]{\frac{(\la+c)^2 - 4d^2}{\la^2 - 1}}, \quad \tilde\delta(2)>0,\
\la\in\C\setminus [-c-2d, 1].
\]
Let the interval $[-c-2d,1]$ be oriented in the positive direction. Then
\[
\tilde\delta_+(\la)=\tilde\delta_-(\la)
\begin{cases} -\I, &\la\in (-1,1),\\
                                 -1, & \la\in(-c+2d, -1),\\
                                  \I, & \la\in (-c-2d,-c+2d).
\end{cases}
\]
Set $\delta(p)=\tilde\delta(\la)$ as $p=(\la, +)\in \Pi_U$ and $\delta(p)=\delta^{-1}(p^*)$.
Then $\delta(p)$ solves the following conjugation problem
\[
\delta_+(p)=\delta_-(p) \begin{cases}
\left|\sqrt{\frac{(\pi(p)+c)^2 - 4d^2}{\pi(p)^2 - 1}}\right|, & p\in\Sigma\cup\Sigma_1,\\
-1,& p\in I\cup I^*.\end{cases}
\]
On the other hand, the transmission coefficient $f(p)=T(p)$ is a single-valued function
on the upper sheet of $\M$ and takes complex conjugated values in symmetric points of $\Sigma$.
Continue $f$ on the lower sheet by $f(p)=T^{-1}(p^*)$, $p=(\la,-)$. Then
$f(p^*)=f^{-1}(p)$, $p\in\M$, and it is a solution bounded at $\infty_\pm$ of the problem
\be\label{jumpforf}
f_+(p)=f_-(p)|T(p)|^2, \quad p\in\Sigma\cup\Sigma_1.
\ee
We also observe that $f(p)$ and $\delta(p)$ are bounded at $\infty_\pm$.
Taking into account  \eqref{propchi} and \eqref{TW} we conclude that the function
$F(p)=\delta(p)f(p)$ solves the problem \eqref{condu}--\eqref{ffff} and $F_+(p)=- F_-(p)$ as
$p\in  I\cup  I^*$. Comparing this with \eqref{DDD} we get
\eqref{oo}. Thus both RH problems provide the same finite-gap solution.

We return now to our case. The function $T(p)$ has a simple pole at $p_0$ on $\Pi_U$, and we 
have to take as $f(p)$ the following function:
 \[
 f(p)=\begin{cases} B(p,p_0)T(p),\quad p\in\Pi_U,\\ f^{-1}(p^*), \quad p\in\Pi_L,\end{cases}
 \] 
where 
 \[
B(p,p_0)=\exp\left(\int_{-c-2d}^p \omega_{p_0 p_0^*}\right)
\] is the Blaschke factor (cf.\ \cite{Teschl1}). Since $|B(p,p_0)|=1$ as $p\in \Sigma\cup\Sigma_1$,
then $f(p)$ solves the jump problem \eqref{jumpforf}. But unlike in the previous case, the Blaschke factor 
has a jump along $I\cup I^*$, 
 \[
 B_+(p,p_0)=B_-(p,p_0)\exp\left(\int_{\mathfrak b}\omega_{p_0,p_0^*}\right).
 \]
 Thus
 \[
 \Delta +\Delta_1=\pi + \frac{1}{\I}\int_{\mathfrak b}\omega_{p_0,p_0^*},
 \]
 or, taking into account \eqref{zetasv},
\be\label{zetasvf} \ul z_1(n,t)=\ul z(n,t) -\frac{1}{2\pi\I}\int_{\mathfrak b} \omega_{p_0p_0^*}\quad (\mbox{mod}\  1).\ee
To formulate the result, recall that all objects introduced in Sections~\ref{sec:model} and \ref{sec:el}, namely, $A(p)$, $\tau$, $\Lambda$, $U$, $\Xi$, $Y$, $\Delta$, $\tilde a$, $\tilde b$ do
not depend on $\xi$ for $\xi \in (\xi_{cr,1}^\prime, \xi_{cr}^\prime)$. Respectively, the finite-gap solutions constructed in \eqref{aee}, \eqref{bee} or \eqref{aee1}, \eqref{bee1} do not depend on $\xi$. Formula \eqref{zetasvf} implies that we can represent them via 
\begin{align}\label{aee4}
\hat a^2(n,t)&=\tilde a^2\,\frac{\theta(\underline{z}(n-1,t))\theta(\underline{z}(n+1,t))}{\theta^2(\underline{z}(n,t))}, \\
\label{bee4}
\hat b(n,t)&=\tilde b +\frac{1}{Y}\,\frac{\pa}{\pa w}\log\left(\frac{
\theta\big(\underline z(n-1,t) + w\big)}{\theta\big(\underline z(n,t) + w\big)}\right)\Big|_{w=0},
\end{align}
where the argument of the theta-function undergoes a phase shift due to the presence of the eigenvalue,
\be\label{zee4} \underline{z}(n,t)=A(\infty_+)- n\frac{\Lambda}{2\pi\I}-t\frac{U}{2\pi\I} +\frac{\tau}{2}
+\frac{\Delta}{2\pi}-\frac{\Lambda}{2\pi\I}- \Xi, \quad  \xi\in (\xi_{cr,0},\xi_{cr}^\prime),
\ee
and \be\label{ze5}\underline{z}(n,t)=A(\infty_+)- n\frac{\Lambda}{2\pi\I}-t\frac{U}{2\pi\I} +\frac{\tau}{2}
+\frac{\Delta}{2\pi}-\frac{\Lambda}{2\pi\I}- \Xi -\frac{1}{2\pi\I}\int_{\mathfrak b} \omega_{p_0p_0^*}, \ \xi\in (\xi_{cr,1}^\prime, \xi_{cr,0}).
\ee

\begin{theorem}  \label{theor3} 
In the domain $\xi_{cr,0} t < n < \xi_{cr}^\prime t$, the solution $\{a(n,t), b(n,t)\}$ of the problem \eqref{tl}, \eqref{ini}--\eqref{cond15} is asymptotically close as $t\to\infty$ to
the two band solution  $\{\hat a_q(n,t), \hat b(n,t)\}$  constructed by \eqref{aee4}--\eqref{zee4}.
In the domain $\xi_{cr,1}^\prime t < n < \xi_{cr,0} t$, the solution $\{a(n,t), b(n,t)\}$ of the problem \eqref{tl}, \eqref{ini}--\eqref{cond15} is asymptotically close as $t\to\infty$ to
the two band solution  $\{\hat a_q(n,t), \hat b(n,t)\}$   constructed by \eqref{aee4}, \eqref{bee4}, \eqref{ze5}.
\end{theorem}

If $d=\frac{1}{2}$ such that the spectra of the background operator are of equal length, then the solution of \eqref{tl}, \eqref{ini}--\eqref{cond15} is close to the periodic Toda lattice solution, which undergoes 
a phase shift if discrete spectrum is present, as shown in \cite{vdo}.

\section{Asymptotics of the solution for \texorpdfstring{$\xi > \xi_{cr}$}{xi > x,cr} and \texorpdfstring{$\xi < \xi_{cr,1}$}{xi < xi,cr,1}}
\label{sec:lr}

Let us consider the domain $\xi > \xi_{cr}$ first, where we study the right RH problem I,  \eqref{eq:jumpcond}, \eqref{polecond}, IV. The signature table of
$\re \Phi(p)$  in this case is depicted in Fig.~\ref{fig:signRePhi5}.
\begin{figure}[ht]
\begin{picture}(6,2.3)

\put(-0.5,1){$+$}
\put(3.15,1.5){$-$}
\put(3.15,0.5){$-$}

\linethickness{0.2mm}
\curve(0,0, 0.3,1, 0,2)
\put(0.2,1.8){$\re \Phi =0$}

\linethickness{0.2mm}
\put(1,1){\line(1, 0){1.5}}
\put(1,1){\circle*{0.1}}
\put(2.5,1){\circle*{0.1}}
\put(1.7,0.6){ $I_1$}

\linethickness{0.6mm}
\put(4,1){\line(1, 0){2}}
\put(4,1){\circle*{0.1}}
\put(6,1){\circle*{0.1}}
\put(4.9,0.6){ $\Sigma$}
\put(4.5,1.15){$\re \Phi =0$}

\end{picture}
\caption{Signature table of $\re \Phi(p, \xi)$ for $\xi > \xi_{cr}$ on $\Pi_U$}\label{fig:signRePhi5}
\end{figure}
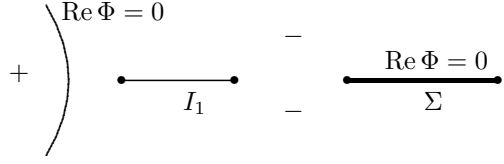
Here $\Phi$ serves as the $g$-function itself. Step 1 is the same with $I_3=I_1=[-c-2d,-c+2d]$, which means
that we switch from the initial Riemann surface $\M$ to the Riemann surface $\hat\M$ of the function $\sqrt{\la^2-1}$.
Step 2 is done on the domain $\Omega\subset\mathfrak D\cap \{p\in\Pi_U \mid \re\Phi(p)<0\}$ as depicted in Figure~\ref{fig:S2}.
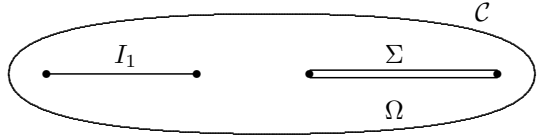
\begin{figure}[ht]
  \begin{picture}(7,2.2)

  \linethickness{0.2mm}
  \put(0.5,1){\line(1,0){2}}
  \put(0.5,1){\circle*{0.1}}
  \put(2.5,1){\circle*{0.1}}
  \put(1.4,1.15){$I_1$}

  \put(4,1.05){\line(1,0){2.5}}
  \put(4,0.95){\line(1,0){2.5}}
  \put(4,1){\circle*{0.1}}
  \put(6.5,1){\circle*{0.1}}
  \put(5,1.15){$\Sigma$}

  \linethickness{0.1mm}
  \closecurve(0,1, 3.5,1.8, 7,1, 3.5,0.2)
  \put(5, 0.4){$\Omega$}
  \put(6.2,1.7){$\mathcal{C}$}

  \end{picture}
  \caption{The lens contour for Step 2 on $\Pi_U(\xi)$.}\label{fig:S2}
\end{figure}
Taking into account \eqref{pruff}, we obtain that the jump matrix on $I_1\cup I_1^*$ is simply the identity matrix. Thus on $\hat\M$ we have the following conjugation problem: $m_+^2(p)=m_-^2(p)v^2(p)$  with
\[
v^2(p) = \begin{cases}
\id, & p \in \Sigma, \\
 \begin{pmatrix}
1 & 0 \\
-R(p) \E^{2 t \Phi(p)} & 1
\end{pmatrix}, & p \in \mathcal{C},\\[4mm]
 \begin{pmatrix}
1 & -R(p^*) \E^{-2 t \Phi(p)} \\
0 & 1
\end{pmatrix}, & p \in \mathcal{C}^*.
\end{cases}
\]
These transformations did not affect the asymptotical behavior of the initial solution.
The jumps on $\mathcal{C} \cup \mathcal{C}^*$ are close to the identity matrix
with a difference that is exponentially small with respect to $t$. Clearly, the unique solution of the model problem $m^{mod}_+(p)=m^{mod}_-(p)v^{mod}(p)$ on $\Sigma$ with
standard normalization and symmetry conditions is the unit vector.
Taking into account Lemma~\ref{asypm}, we get $a(n,t)=\frac{1}{2}+o(1)$
and $b(n,t)=o(1)$ when $t \to \infty$ in the region $\frac{n}{t}> \xi_{cr}$.

Studying the left RH problem I, \eqref{eq:jumpcond2}--III in the domain $\frac{n}{t} < \xi_{cr,1}$,
we obtain in the same manner that $a(n,t)=d+o(1)$ and $b(n,t)=-c + o(1)$ when $t \to \infty$.

\appendix
\section{Proof of Lemma \ref{abelint}}

We start by proving items (i) and (iv), i.e., we first prove that the points $\gamma(\xi)\in (-c-2d, -c +2d)$ and $\mu(\xi)\in (\gamma(\xi),-1)$ satisfying \eqref{eqas}--\eqref{zerocond} can be chosen uniquely, and that  $\gamma(\xi)$ moves continuously to the right from $\gamma(\xi_{cr})=-c-2d$ to 
$\gamma(\xi_{cr}^\prime)=-c+2d$ with respect to decreasing $\xi$.
First of all note that if
 $\gamma\in I_1$, then  $\mu=\mu(\gamma)\in (\gamma, -1)$ is defined as
\be\label{defmu}
\mu(\gamma)=\int_{\gamma}^{-1}\frac{\la (\la - \gamma)}{R^{1/2}(\la,\gamma)}d\la\,\left(\int_{\gamma}^{-1}\frac{ \la - \gamma}{R^{1/2}(\la,\gamma)}d\la\right)^{-1}.
\ee
Evidently, $\mu(\gamma)$ is a continuous function of $\gamma$ and by the mean value theorem
\be\label{triv}
\gamma<\mu(\gamma),\quad \forall \gamma \in (-c-2d, -c+2d).
\ee
Now consider $\mu$ as a function of $\xi$ defined via \eqref{eqas} and insert it into \eqref{zerocond}. Then
\[
F(\gamma,\xi):=-\int_\gamma^{-1}\frac{\sqrt{\la - \gamma}\big(\la + \xi +\frac{\gamma + c+2d}{2}\big)}{\sqrt{(\la^2 - 1)(\la + c+2d)}}d\la
\]
satisfies $F(\gamma(\xi),\xi)\equiv 0$. Thus $\gamma(\xi)$ is an implicitly given function and
\[
\frac{\partial F}{\partial \gamma}\frac{ d\gamma}{d\xi} + \frac{\partial F}{\partial \xi}=0.
\]
Since
\[
\frac{\partial F}{\partial \gamma}=-\frac{1}{2}\left(\xi +\frac{3\gamma +c+2d}{2}\right)\int_\gamma^{-1}\frac{d\la}{R^{1/2}(\la,\gamma)}, \quad
\frac{\partial F}{\partial \xi}=\int_\gamma^{-1}\frac{\la - \gamma}{R^{1/2}(\la,\gamma)}d\la,
\]
and $R^{1/2}(\la,\gamma)$ does not change its sign on the interval $(\gamma, -1)$, then
\be\label{derivg}
\frac{ d\gamma}{d\xi}=\frac{4\mathcal{K}(\xi)}{2\xi + 3\gamma + c + 2d},
\ee
where
\[
\mathcal{K}(\xi)= \int_{\gamma(\xi)}^{-1}\frac{\la - \gamma(\xi)}{R^{1/2}(\la,\gamma(\xi))}d\la\,
\left(\int_{\gamma(\xi)}^{-1}\frac{d\la}{R^{1/2}(\la,\gamma(\xi))}\right)^{-1}>0
\]
in the region under consideration. We want to show that
\be\label{defef}
f(\xi):=2\xi +3\gamma(\xi) + c + 2d<0,  \text{ for }  \xi\in (\xi_{cr}^\prime,\ \xi_{cr}).
\ee
First observe that
\be\label{leftcr}
\gamma(\xi_{cr})=-c-2d.
\ee
Namely, for any $\gamma$, the function $\mu(\gamma)$ is defined by \eqref{defmu}, which for $\gamma=-c-2d$ is equal
to
\[
\mu(-c-2d)=\int_{-c-2d}^{-1}\frac{\la}{\sqrt{\la^2 - 1}}d\la\,\left( \int_{-c-2d}^{-1}\frac{d\la}{\sqrt{\la^2 - 1}}\right)^{-1}=- \xi_{cr,2}
\]
by \eqref{xicrit2}.
On the other hand, for $\mu=- \xi_{cr}$, $\gamma=-c-2d$, and $\xi= \xi_{cr}$, \eqref{eqas} is also satisfied. Thus \eqref{leftcr} is true and $f(\xi_{cr})=2(\xi_{cr} -c-2d)$. Moreover, \eqref{xicrit2} and $c+2d>1$ imply
\be\label{ineq7}2 \xi_{cr}\leq c+2d+1.
\ee

Since $c+2d+1<2c+4d$  by \eqref{cond7}, then $f(\xi_{cr})<0$ where $f$ is defined by \eqref{defef}. Moreover, the function $\gamma$ is continuously differentiable with respect to $\xi$ at least in the
right vicinity of $\xi_{cr}$. In fact, it will be continuously differentiable up to the first point $\xi$ where $f(\xi)=0$. But if $f(\xi)=0$ then $3\gamma +c+2d=-2\xi$. On the other hand, by \eqref{eqas}
$-2\xi=\gamma + 2\mu + c +2d$. Thus at $\xi_0$ where $f(\xi_0)=0$, one obtains $\mu=\gamma$ which contradicts \eqref{triv}. Therefore $\frac{d\gamma}{d\xi}<0$, starting at the point $\gamma=-c-2d$ where $\xi=\xi_{cr}$ and at least up to the point $\gamma=-c+2d$ where $\xi=\xi_{cr}^\prime$.

Now we prove (ii), i.e., we prove representations \eqref{cond37}--\eqref{cond333}. Let $h=h(\xi)$  be defined by \eqref{cond36}. 
Given $\gamma$, $\mu$ and $h$, we observe that $\nu_1$ and $\nu_2$ are zeros of the polynomial
 \[
 p(\la)= (\la-\nu_1)(\la-\nu_2) = (\la - \gamma)(\la - \mu) - \xi(\la - h)
 \]
 with real-valued coefficients. These zeros cannot be complex conjugated, because  \eqref{cond37}, \eqref{cond35} (b), and \eqref{zerocond} imply
 \[
 \int_{\gamma}^{-1}\frac{p(\la)d\la}{R^{1/2}(\la,\gamma)}=0,
 \]
moreover, this formula implies that at least one zero belongs to the interval $(\gamma, -1)$.
Condition \eqref{cond333} is true due to the asymptotical behavior of the l.h.s. of \eqref{cond37} and  function $\omega(\la,\xi)$.

To prove \eqref{deri} we observe that
\be\label{idd}
-\frac{\pa}{\pa\xi}\frac{(\la - \mu(\xi))(\la-\gamma(\xi))}{R^{1/2}(\la,\gamma(\xi))} +\omega(\la,\xi)=
  \frac{2\mu^\prime(\la-\gamma) + \gamma^\prime(\la-\mu)}
 {2R^{1/2}(\la,\gamma)} +\frac{\la-h}{R^{1/2}(\la,\gamma)}.\ee
 By \eqref{eqas}, $2\mu^\prime + \gamma^\prime=-2$, therefore \eqref{idd} is in turn equivalent to
 \be\label{zhut2}
2\mu^\prime(\xi)\gamma(\xi)+ \gamma^\prime(\xi) \mu(\xi) =- 2h(\xi).
\ee Hence to prove \eqref{deri} amounts to proving \eqref{zhut2}.
From \eqref{derivg} we have
\be\label{primeg}
\gamma^\prime(\xi)=\frac{4 b(\xi) - 4\gamma(\xi)}{2\xi + 3\gamma(\xi) + c + 2d},
\ee
where $b(\xi)$ is defined by \eqref{cond40}.  
On the other hand, \eqref{eqas} implies
$2\mu^\prime= -2 -\gamma^\prime$ and $2\mu=-2\xi-\gamma+c+2d$.
Substituting this into the l.h.s.\ of \eqref{zhut2} and using  \eqref{primeg} yields
$$
2\mu^\prime \gamma +\gamma^\prime \mu 
=-\frac{\gamma^\prime}{2}\left( 3\gamma + 2\xi + c + 2d\right) - 2\gamma=-(2b - 2\gamma) - 2\gamma=-2b.
$$
By \eqref{cond36} $b(\xi)=h(\xi)$, which proves \eqref{zhut2}.

\section{Uniqueness for the model problem}

In this appendix we prove uniqueness for the solution of the model problem
(4.9)--(4.13), which admits weak singularities at two points on the jump contour.
\begin{lemma} Let $m(p)=(m_1(p), m_2(p))$, $p\in\M(\xi)$, be a solution of the problem (4.9)-(4.13), which is holomorphic in $\M(\xi)\setminus \mathcal L$,
$\mathcal L:=\left(\Sigma_1(\xi)\cup\Sigma_2\cup I_3\cup I_3^*\right)$, and has continuous limits as $p$ approaches any 
point of the contour $\mathcal L$ with the exception of the two branch points $E_2:=(\gamma(\xi),\pm)$ and $E_1:=(-c-2d,\pm)$.
Let $m(p)=O((p - E_j)^{-1/4})$ as $p\to E_j$, $j=1,2$. The solution with such properties is unique.
\end{lemma}
\begin{proof}
Let $\tilde m=(\tilde m_1, \tilde m_2)$ and 
 $\breve m=(\breve m_1, \breve m_2)$ be two solutions satisfying all conditions of the lemma.
Consider them as functions on the Riemann surface $\hat\M$ of the function $\sqrt{\lambda^2 - 1}$.  Then the contour $\Sigma_1(\xi)$ transforms onto two contours: $I_1=[E_1, E_2]$ on the upper sheet of $\hat M$, oriented in positive direction, and
$I_1^*$ on the lower sheet, oriented in negative direction. The jump matrix for $\tilde m$ and $\breve m$ on $I_1$ will be $v(p)=\I\sigma_1$, and on $I_1^*$  $v(p)=-\I\sigma_1$.

Consider the matrix
$$S(p)=\begin{pmatrix} \tilde m_1& \tilde m_2\\
\breve m_1 & \breve m_2\end{pmatrix}$$ and let $s(p)=\det S(p)$. Since $\det v(p)=1$ on $\Sigma\cup I\cup I^*\cup I_1\cup I_1^*$  then $s(p)$ has no jumps on $\hat \M$ and is a holomorphic
function on this surface except of four points $(E_j,\pm)$, $j=1,2$. In these points  $s(p)$ has isolated singularities of order $O((\pi(p) - \pi(E_j))^{-1/2})$. Therefore these points on the upper and lower sheets of $\hat \M$ are removable singularities. We conclude that $s(p)$ is holomorphic on $\hat \M$ and bounded at $\infty_\pm$ due to normalization conditions for $\tilde m$ and $\breve m$. By the Liouville theorem $s(p)\equiv const$ on $\hat\M$. The
 symmetry condition implies
$s(p)+s(p^*)=0$, respectively at the branch points of $\hat M$ we have $s(-1)= s(1)=0$, that is $s(p)=0$ as $p\in\hat \M$.

Therefore $(\hat m_1(p),\hat m_2(p))=c(p) (\tilde m_1(p), \tilde m_2(p))$, where a scalar function $c(p)$ has no jumps on $\hat\M$, 
moreover, the normalization condition implies $\lim_{p\to\infty_\pm}c(p)=1$. Respectively for any two solutions $\breve m$ and $\tilde m$ we have in fact $\hat m(\infty_\pm)=\tilde m(\infty_\pm)$. Therefore, to prove the uniqueness of the solution of the RH problem under consideration, it is sufficient to prove that the only  solution $m(p)$, satisfying the  jump condition (4.9), symmetry condition, "vanishing  condition" $m(\infty_\pm)=0$, which has   "weak singularities" of order $O((p - E_j)^{-1/4})$ at $E_j$, is the  trivial solution.

Let $m(p)=(m_1(p), m_2(p))$ be a solution of this "vanishing problem".
Consider the weight 
$$
d\Omega=\frac{\I d\lambda}{\sqrt{\lambda^2 - 1}},\quad p=(\lambda,+)\in \M(\xi),
$$ 
on  the initial  Riemann surface $\M(\xi)$, which corresponds to the model problem. Let $\mathcal C$ be a closed contour on the upper sheet oriented counterclockwise. Since the function $f(p):=m(p)m^\dagger(\overline{p^*})$ 
has the behavior $f(p)=O\left((p - E_j)^{-1/2}\right)$ at the edges of $\Sigma_1(\xi)$, it is integrable and by the residue theorem
$$0=\int_{\mathcal C} f(p)d\Omega=\int_{\Sigma_1(\xi)}f(p)d\Omega + \int_{\Sigma_2}f(p)d\Omega = J_1 + J_2.
$$ 
The integrals along the upper and lower sides of
$I$ cancel each other due to the jump condition (4.9).
Since $\sqrt{\lambda^2 - 1}<0$ as $p\in\Sigma_1^u(\xi)\cup \Sigma_1^l(\xi)$ and $\I d\Omega>0$ as $p\in\Sigma_1^u(\xi)$ then by  (4.9)
$$J_1=\I\int_{\Sigma_1^u(\xi)} (|m_{1,-}(p)|^2 + |m_{2,-}(p)|^2)d\Omega -\I \int_{\Sigma_1^l(\xi)} (|m_{1,-}(p)|^2 + |m_{2,-}(p)|^2)d\Omega\geq 0.$$
On $\Sigma$ we have $d\Omega>0$ and $m_+(p)=m_-(p)=m(p)$, therefore
$$J_2=\int_{\Sigma}\|m(p)\|^2d\Omega\geq 0.$$ Thus
$m(p)=0$ on $\Sigma_1(\xi)$ and $\Sigma$, except of $E_1$ and $E_2$. Since both components of $m(p)$ have no jump in a neighborhood of $E_1$, this point is an isolated singularity, and since $m(p)= O((p-E_j)^{-1/4})$, then in fact $m(p)$ is bounded near $E_1$. 

Introduce  now a new weight $$d\Omega_1(p) =\frac{ \I d\lambda}{\sqrt{(\lambda +1)(\lambda-E_1)}},\quad p=(\lambda, +)\in\M(\xi).$$ We observe that $d\Omega_1>0$ as $p\in I$.
Consider (note, that here we integrate $m^\dagger(\overline{p})$, not $m^\dagger(\overline{p^*})$)
$$
0=\int_{\mathcal C} m(p)m^\dagger(\overline{p})d\Omega_1(p)=\int_{\Sigma_1(\xi)\cup \Sigma}m(p)m^\dagger(\overline{p})d\Omega_1 + J=J,
$$
where
$$
J=2\cos y\int_{I}(|m_{1,-}|^2 + |m_{2,-}|^2)d\Omega_1,
$$ 
with $y=2tB + \Delta.$
Since $J=0$ then for those $t$ for which $\cos(2tB+\Delta)\neq 0$, we have $m_-(p)=m_+(p)=0$ on $I$.  The same consideration as above then shows that $m$ has an isolated removable singularity at $E_2$ too.
Since $m(p)=0$ on $I$ and, respectively, on $I^*$, $m(p)\equiv 0$.

\end{proof}

\bigskip
\noindent{\bf Acknowledgments.} We thank Spyros Kamvissis, Irina Nenciu, and Dmitry Shepelsky for discussions on this topic. I.E. is indebted to the Department of Mathematics at the
University of Vienna for its hospitality and support during the winter of 2014, where
some of this work was done.

\end{document}